\documentclass[draftcls,onecolumn,12pt]{IEEEtran}
\usepackage{color}
\usepackage{verbatim}
\usepackage{amsfonts}
\usepackage{amssymb}
\usepackage{stfloats}
\usepackage{cite}
\usepackage{graphicx}
\usepackage{psfrag}
\usepackage{amsmath}
\usepackage{array}
\usepackage{epstopdf}
\usepackage{authblk}
\usepackage{graphicx} 
\usepackage{amsthm} 
\usepackage{lipsum}
\usepackage{verbatim} 
\usepackage{authblk}
\usepackage{mathtools}
\usepackage{cuted}
\usepackage[lined,boxed,ruled]{algorithm2e}
\usepackage{algpseudocode}
\usepackage{framed} 
\usepackage{subfigure}
\usepackage{soul}

\newtheorem{Lemma}{Lemma}
\newtheorem{Theorem}{Theorem}

\title{\LARGE{Progressive Feature Transmission for Split Inference\\ at the Wireless Edge}}

\author{
 Qiao~Lan, Qunsong~Zeng, Petar~Popovski, Deniz~G\"und\"uz, and Kaibin~Huang
 \thanks{Q. Lan, Q. Zeng, and K. Huang are with Department of Electrical and Electronic Engineering at The University of Hong Kong, Hong Kong (Email: \{qlan, qszeng, huangkb\}@eee.hku.hk).  D. G\"und\"uz is with Department of Electrical and Electronic Engineering at Imperial College London, London, UK (Email: d.gunduz@imperial.ac.uk). P. Popovski is with Department of Electronic Systems at Aalborg University, Aalborg, Denmark (Email: petarp@es.aau.dk). Corresponding author: K. Huang. }
 }

\makeatletter
\newcommand{\removelatexerror}{\let\@latex@error\@gobble}
\makeatother

\IEEEoverridecommandlockouts

\begin{document}

\maketitle

\vspace{-12mm}
\begin{abstract}
In \textit{edge inference}, an edge server provides remote-inference services to edge devices. This requires the edge devices to upload  high-dimensional features of data samples over resource-constrained  wireless channels, which creates a communication bottleneck. The conventional solution of feature pruning requires that the device has access to the inference model, which is unavailable in the current scenario of split inference. To address this issue, we propose the \emph{progressive feature transmission} (ProgressFTX) protocol, which minimizes the overhead by progressively transmitting  features until  a target confidence level  is reached. The optimal control policy of the protocol to accelerate inference is derived  and it comprises two key operations. The first is  \emph{importance-aware feature selection} at the server, {for which it is shown to be optimal to select the most important features,} characterized by the largest  discriminant gains of the corresponding  feature  dimensions. The second is \emph{transmission-termination  control} by the server for which the optimal policy is shown to exhibit a threshold  structure. Specifically, the transmission is stopped when the incremental uncertainty  reduction by further feature transmission  is outweighed by its  communication cost. The indices of the selected features and transmission decision are fed back to the device in each slot. The optimal policy is first derived for the tractable case of linear classification and then extended to the more complex case of classification using a \emph{convolutional neural network}. Both Gaussian and fading channels are considered. Experimental results are obtained for both a statistical data model and a real dataset. It is seen that ProgressFTX can substantially  reduce the communication latency compared to conventional  feature pruning and random feature transmission.

\end{abstract}

\vspace{-4mm}

\section{Introduction}

Recent years have witnessed the extensive deployment of \emph{artificial intelligence} (AI) technologies at the network edge to gain fast access to and processing of mobile data. This gives rise to two active research challenges: (1) \emph{edge learning}~\cite{CMZ2021Arxiv,GX2020CM}, where data are used to train large-scale AI models via distributed machine learning; and (2) \emph{edge inference}~\cite{CMZ2021Arxiv,Wang2020Tutorial}, which is the theme of this work and deals with operating of such models at edge servers to provide remote-inference services that enable emerging mobile applications, such as e-commerce or smart cities. For instance, a large-scale remote classifier (e.g., Google or Tencent Cloud) can power  a device to  discriminate hundreds of object classes. To reduce communication overhead and protect data privacy, a typical edge-inference algorithm involves a device extracting features from raw  data and transmitting them  to an edge server for inference. To improve the communication efficiency, we propose and optimize a simple protocol, termed progressive feature transmission (ProgressFTX). 

The state-of-the-art edge learning algorithms build upon an architecture termed \emph{split inference}~\cite{Zhou2020IoTJ,Niu2019Infocom,Chen2020TWC,Zhang2020CM,Zhang2020ICC,Deniz2020SPAWC,Deniz2021JSAC}, {in which the model is partitioned into \emph{device} and \emph{server} sub-models}~\cite{Niu2019Infocom,Zhang2020CM}. Using the device sub-model, an edge device extracts features from a raw data sample and uploads them to a server, which then uses  these features to compute an inference result and sends it back to the device. The device-server splitting of the computation load can be adjusted by shifting the splitting point of the model. It is proposed in~\cite{Chen2020TWC} to adapt the splitting point to the communication rate so as to meet a latency requirement. Split inference can support adaptive model selection via compression of a full-size \emph{deep neural network} (DNN) model into numerous reduced versions~\cite{Niu2019Infocom}.  The hostility of the wireless link between the device and server introduces a communication bottleneck for the split inference.  This issue is addressed by the framework of joint source-and-channel coding developed  over a series of works~\cite{Zhang2020CM,Zhang2020ICC,Deniz2020SPAWC,Deniz2021JSAC}, featuring the use of an autoencoder pair of encoder and decoder as device and server sub-models, jointly trained to simultaneously perform inference and efficient transmission. In view of prior works, there is a lack of a communication-efficient solution that can adapt to a practical time-varying channel. Particularly, with the exception of~\cite{Deniz2020SPAWC,Deniz2021JSAC}, the aforementioned joint source-and-channel coding approaches assume a static channel and require model retraining whenever the propagation environment changes significantly. On the other hand, the joint source-and-channel coding schemes addressing fading channels are applicable to analog transmission only~\cite{Deniz2020SPAWC,Deniz2021JSAC}. While feature transmission, as opposed to raw data uploading, leads to substantial overhead reduction, the communication bottleneck still exists due to the high dimensionality of features required for satisfactory inference performance~\cite{Xu2018NatureElectron}. For instance, GoogLeNet, a celebrated \emph{convolutional neural network} (CNN) model, generates 256 feature maps at the output of a convolutional layer, resulting in a total of $2\times 10^{6}$ real coefficients~\cite{Szegedy2015CVPR}. One solution for overcoming the communication bottleneck is to prune features according to their heterogeneous importance levels~\cite{Zhang2020CM,Deniz2020SPAWC,Guo2020AAAI}. There are several representative approaches for importance-aware feature pruning in DNN models. The first approach is to evaluate the importance of a feature by observing the effect of its removal on the inference performance~\cite{Guo2020AAAI,Zhang2018NeurIPS,Molchanov2019CVPR}. The second approach devises  explicit importance measures, including the sum of cross-entropy loss and reconstruction error~\cite{Zhang2018NeurIPS}, Taylor expansion of the error induced by pruning~\cite{Molchanov2019CVPR}, and symmetric divergence~\cite{Saon00NeurIPS}. These algorithms outperform the na\"{i}ve approach that assigns  a higher importance to a larger parametric magnitude. The last approach, known as learning-driven coding, relies on an auto-encoder to intelligently identify and encode discriminant features to improve the communication efficiency~\cite{Zhang2020CM,Deniz2021JSAC}. The proposed ProgressFTX protocol targets split inference and aims at achieving a higher  efficiency  than the existing one-shot feature selection/pruning by exploiting, besides importance awareness, the stochastic control according to the channel state. 

Here we consider the scenario in which split inference is to be deployed for classification tasks. For a given data sample, its classification accuracy improves as more features are sent to the server. This fact naturally motivates  ProgressFTX to reduce the communication overhead. Specifically, based on the protocol, a device progressively transmits selected subsets of features until a target classification accuracy is reached, as informed by the server, or the expected communication cost becomes too high. The principle of progressive transmission also underpins \emph{automatic repeat request} (ARQ), a basic mechanism for reliability in wireless networks.  The reliability  is achieved by repeating the transmission of a data packet until it is successfully received.  Among the established communication protocols, ProgressFTX is most similar to \emph{hybrid ARQ} (HARQ) with incremental redundancy~\cite{Lin1982TCOM}. Integrated with forward error correction, HARQ sequentially transmits multiple punctured versions of the encoder output for the same input data packet. Upon retransmission, the  reliability of the received packet is checked by combining all received versions and performing error detection on the combined result, thereby progressively improving the reliability. Despite the similarity, ProgressFTX differs from HARQ in  two main aspects. First, ProgressFTX is a cross-disciplinary design aiming at achieving both a high classification accuracy and low communication overhead, while HARQ targets communication reliability. Second, the incremental redundancy in ProgressFTX differs from that for HARQ and refers to the use of an increasing number of features to improve classification accuracy. For example, using $16$ feature maps  from the MNIST dataset leads to an accuracy of $93\%$ and additional $4$ feature maps boost the accuracy to $98\%$. As a result, the HARQ operations of punctured error control coding, combining,  and error detection are replaced with feature extraction, feature cascading, and classification, respectively, in the context of ProgressFTX. 

The contribution of this work is the  proposal of the ProgressFTX protocol and its optimal control. While a transmitted feature reduces the uncertainty in classification, its transmission incurs communication cost. This motivates the optimization of  ProgressFTX, which is formulated as  a problem of   optimal stochastic control with  dual objectives of minimizing both the uncertainty and communication cost. The problem is solved analytically for the case of linear classifier and the results are extended to a general model (e.g., CNN).

\begin{itemize}
	\item \emph{Optimal ProgressFTX for linear classifiers:}   The optimization problem  is decomposed into two sub-problems that reflect the two main operations at the server: optimal feature selection and optimal stopping.  The key idea to get tractable solutions is to approximate the commonly used entropy of posteriors by constructing a sufficiently tight upper bound as a convex function of the differential Mahalanoibis distance from a feature vector to each pair of class centroids of  the data distribution. With this approximation, the provably optimal strategy of the server is to follow the order of decreasing feature importance as measured by its entropy. The optimal stopping problem is solved for two channel types, Gaussian and fading. The optimal policies for both are derived in closed form and are observed to exhibit a simple threshold based structure. The threshold based polices  specify  criteria for stopping ProgressFTX when the incremental reward (i.e., uncertainty reduction) by further feature transmission is too small,  or the expected communication cost outweighs the reward.    

    \item \emph{Optimal ProgressFTX  for CNN classifiers}: The operations of importance-aware feature selection  and optimal stopping for CNN classifiers are more complex than the case of linear classification due to the complexity of CNNs, rendering an analytical approach intractable. We tackle the challenge by designing and numerically evaluating two practical algorithms. First, we evaluate the importance of a feature map using the gradients of associated model parameters  generated in the process of training the inference model. Second, we advocate the use of a low-complexity regression model trained to predict the incremental inference uncertainty. The model with close-to-optimal performance allows finding the optimal stopping time of ProgressFTX by a simple linear search.

\end{itemize}

Experiments using both statistical and real datasets have been conducted. The results demonstrate that ProgressFTX  with importance-aware feature selection and optimal stopping can substantially reduce the number of channel uses when benchmarked against the schemes of conventional one-shot feature compression or ProgressFTX with random feature selection.

\vspace{-4mm}

\section{Models and Metrics}
\label{sec: models_and_metrics}

We consider the edge inference system from Fig.~\ref{fig: system-diagram}, where local data at an edge device are compressed into features and sent to a server  to do a remote inference on a trained model. 

\vspace{-2mm}

\subsection{Data Distribution  Model}
\label{subsec: data_distribution_model}
While CNN-based classification targets a generic distribution, in the case of linear classification, for tractability, we assume the following data distribution. Each $M$-dimensional data sample is assumed to follow a \emph{Gaussian mixture} (GM) \cite{Hastie2009Elements} and is associated with  one of $L$ classes.  We assume that the $L$ classes have uniform priors.  Given the data distribution, the server computes the $N$-dimensional feature space with $N \leq M$ using   \emph{principal components analysis} (PCA) \cite{Hastie2009Elements}.  At the beginning of the edge-inference process from Fig.~\ref{fig: system-diagram}, the server sends the feature space to the device for the purpose of feature extraction: a feature vector is extracted from each sample by projection onto this feature space. This results in an arbitrary feature vector, denoted as $\mathbf{x} \in \mathbb{R}^{N} $, distributed as a GM in the feature space. Each class $\ell$ is a multivariate Gaussian distribution $\mathcal{N}(\mu_\ell,\mathbf{C}_\ell)$ with a centroid $\mu_\ell \in \mathbb{R}^{N}$ and a covariance matrix $\mathbf{C}_\ell \in \mathbb{R}^{N\times N}$. For tractability,  the covariance matrices are assumed identical:  $\mathbf{C}_\ell=\mathbf{C}, \ \forall \ell $, where $\mathbf{C}$ is diagonalized by PCA and known to the server. Then the \emph{probability density function} (PDF) of $\mathbf{x}$ can be written as
\begin{equation}
\label{eq: sysmodel_signal_pdf}
f(\mathbf{x}) = \frac{1}{L}\sum_{\ell=1}^{L} \mathcal{N}\left(\mathbf{x}|  \mu_\ell,  \mathbf{C}\right),
\end{equation}
where $\mathcal{N}\left(\mathbf{x}|\mu_\ell,\mathbf{C}\right)$ denotes the Gaussian PDF with mean $\mu_\ell $ and covariance $\mathbf{C}$.  

\begin{figure}[t]
\centering
\includegraphics[scale=0.18]{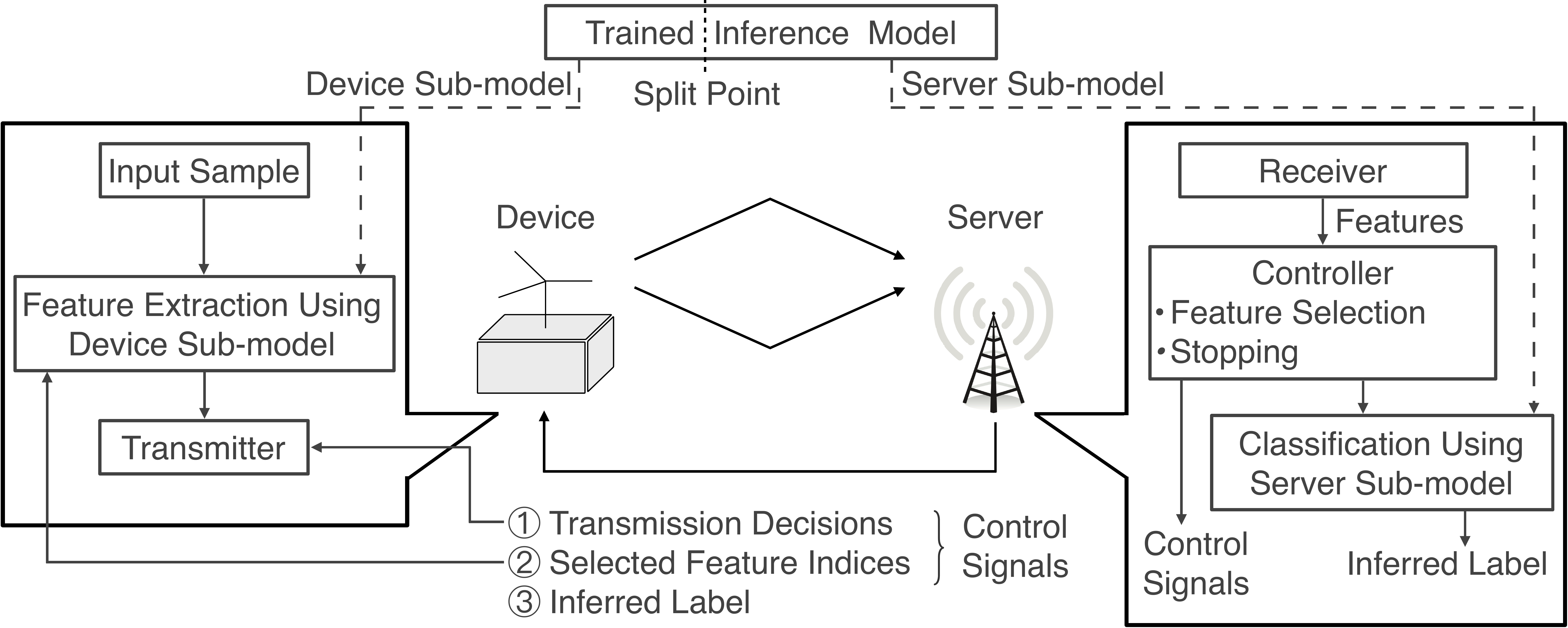}
\vspace{-6mm}
\caption{An edge inference system.}
\label{fig: system-diagram}\vspace{-4mm}
\end{figure}

\vspace{-2mm}

\subsection{Communication Model}
\label{subsec: comm_model}
Consider uplink transmission of feature vectors for remote inference at the server. Each feature is quantized with a sufficiently high resolution of $Q$ bits  such that quantization errors are negligible. We consider both the cases of linear and CNN classification that are elaborated in the next sub-section. In the former,  the  number of features that can be transmitted in one slot is given by $Y_0=\lfloor{\frac{R_0 T}{Q}}\rfloor$ for slot $k$. We refer to $Y_0$ as the transmission rate in features/slot. In the case of CNN classification, the basic unit of transmitted data is a \emph{feature map}, which is a $L_{\sf h}\times L_{\sf w}$ matrix. Correspondingly, the transmission rate (in feature-maps/slot) for a slot can be written as $Y_0 = \lfloor{\frac{R_0 T}{Q L_{\sf h}L_{\sf w}}}\rfloor$. The rounding effect is negligible by assuming a sufficiently high communication rate such that a relatively large number of features or feature maps can be transmitted within each slot (e.g. $>10$). For both linear and CNN classification, the number of transmitted features (or feature maps) is communication-rate dependent and thus varies from slot to slot.

The time of the communication channel is divided into slots, each spanning  $T$ seconds. The device is allocated a narrow-band channel with the bandwidth denoted as $B$ and the gain in slot $k$ as $g_k$. Two types of channels are considered. The first is Gaussian channel with a static  channel gain, where $g_k=g_0$ for all $k$ and is known to both the transmitter and the receiver. The transmit  \emph{signal-to-noise ratio} (SNR) is fixed as $\rho$ and the rate is matched to the channel to be $R_0 = B\log_2(1+\rho g_{0})$.  The second one is fading channel with outage where $g_k$ is an random variable and \emph{independent and identically distributed} (i.i.d.) over slots. Now the transmitter does not know the channel gains $\{g_k\}$, uses a fixed rate $R_0$ and an outage occurs if the channel gain falls  below  $g_0$~\cite{Andrews2005TIT}. Let the outage probability be denoted and defined as $p_{\sf o} = \Pr(g_k < g_0)$. As a final remark, note that the number of features that are available to the server is predictable with the Gaussian model, while it is unpredictable with the fading model. 

\vspace{-2mm}

\subsection{Two Classification  Models}
\label{subsec: classification_model}

\begin{figure}
\centering
\subfigure[Linear Model]{\includegraphics[scale=0.22]{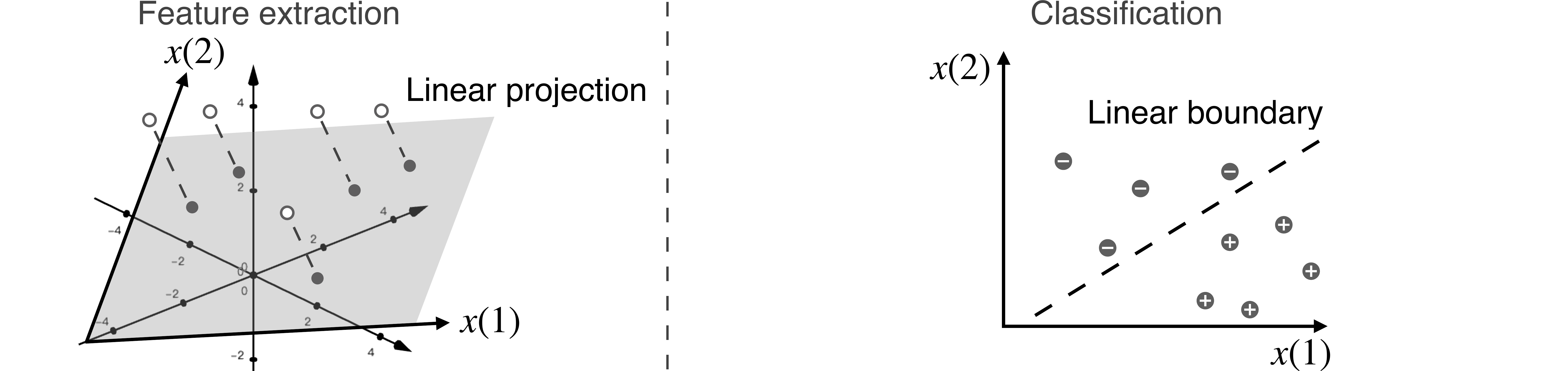}}
\subfigure[CNN Model]
{\includegraphics[scale=0.21]{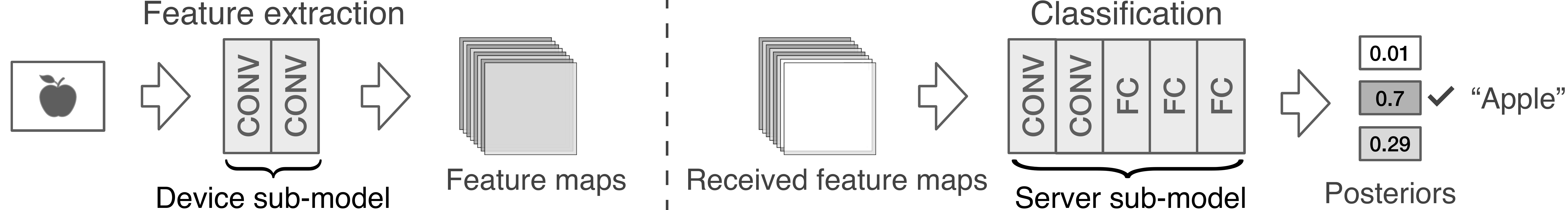}}
\vspace{-4mm}
\caption{Linear and CNN classification models for split inference.}
\label{fig: classification_model}\vspace{-6mm}
\end{figure}

\subsubsection{Linear Classification Model} Consider the progressive transmission of an arbitrary feature  vector $\mathbf{x}$ and let ${x}(n)$ denote the $n$-th feature (coefficient). Let $\mathcal{W}_k \subseteq \mathcal{W}$ denote the cumulative index set of features received by the server by the end of slot $k$, where  $\mathcal{W}=\{1,2,...,N\}$. The corresponding  \emph{partial} feature vector, denoted as $\mathbf{x}_{k}$,  can be defined by setting its coefficients as 
\begin{equation}
\label{eqn: x_k_specs}
{x}_{k}(n) = 
\begin{cases}
{x}(n), & n \in \mathcal{W}_k,\\
0, & \text{otherwise}.  
\end{cases}
\end{equation} 
In each slot, the server estimates the label of the object $\mathbf{x}$ by inputting the partial vector $\mathbf{x}_{k}$ into a classifier. We consider two classification models as described in the sequel. 

The \emph{maximum likelihood} (ML) classifier for the distribution~\eqref{eq: sysmodel_signal_pdf} is a linear one \cite{Hastie2009Elements}, and relies on a classification boundary between a class pair  being  a hyper-plane in the feature space (see Fig.~\ref{fig: classification_model}(a)). Due to uniform prior on the object classes, the ML classifier is equivalent to a \emph{maximum a posterior} (MAP) classifier and the estimated label  $\hat{\ell}$ is obtained as 
\begin{equation}
\label{eqn: linear_classifer}
\hat{\ell} = \mathop{\text{argmax}}\limits_{\ell} \ \Pr\left( \mathbf{x}_{k} | \ell,\mathcal{W}_k \right)
= \mathop{\text{argmax}}\limits_{\ell} \ \exp{\left(-\frac{1}{2}\sum_{n \in \mathcal{W}_k} \frac{\left[x_k(n)-\mu_{\ell}(n)\right]^2}{C_{n,n}}\right)}, 
\end{equation}
where $C_{n,n}$ denotes the $n$-th diagonal entry of the covariance matrix $\mathbf{C}$. The \emph{Mahalanobis distance} refers to the distance  from the sample $\textbf{x}_k$ to the centroid of class-$k$ in the $|\mathcal{W}_k|$-dimensional feature subspace. Given the ground-truth label for $\textbf{x}_k$ being $\ell$, let $z_{\ell}(k)$ denote the half squared Mahalanobis distance: $z_{\ell}(k) \triangleq \frac{1}{2} \sum_{n \in \mathcal{W}_k} \frac{\left[x_k(n)-\mu_{\ell}(n)\right]^2}{C_{n,n}}$. Then the problem of linear classification reduces to one of  Mahalanobis distance minimization: $\hat{\ell} = \mathop{\text{argmin}}_{\ell} \ z_{\ell}(k)$.

\subsubsection{CNN Classification Model}
A CNN model comprises multiple \emph{convolutional} (CONV) layers followed by multiple \emph{fully-connected} (FC) layers. To implement split inference, the model is divided into device and server sub-models as shown in Fig.~\ref{fig: classification_model} (b), which are represented by functions  $f_{\sf d}(\cdot)$ and  $f_{\sf s}(\cdot)$, respectively. Following existing designs in \cite{Zhang2018NeurIPS,Guo2020AAAI}, the   splitting  point is \emph{chosen} right after a CONV layer. Let $N$ denote the number of CONV filters in  the last layer of $f_{\sf d}(\cdot)$, each of which outputs a feature map with height $L_{\sf h}$ and width $L_{\sf w}$. The tensor of all extracted feature maps from  an arbitrary input sample  is denoted as $\mathbf{X} \in \mathbb{R}^{N \times L_{\sf h} \times L_{\sf w}}$, and the $n$-th map as $\mathbf{X}(n)$. Let $\mathbf{X}_k$ denote the tensor of cumulative  feature maps received by the server by slot $k$, and $\mathcal{W}_k$ the corresponding index set of feature maps. Then $\mathbf{X}_k$ can be related to $\mathbf{X}$ as
\begin{equation}
\mathbf{X}_k(n)=
\begin{cases}
\mathbf{X}(n), & n \in \mathcal{W}_k,\\
\mathbf{0},  & \text{otherwise}.
\end{cases}
\end{equation}
We define the system state for this generic model as $\theta_k=(\mathbf{X}_k,\mathcal{W}_k)$. In slot $k$, the server sub-model can compute  the  posteriors  $\Pr(\ell|\mathbf{X}_k)$ for the input $\mathbf{X}_k$  using the  forward-propagation algorithm~\cite{Wang2020Tutorial}. Then, the label $\hat{\ell}$ is estimated by posterior maximization: $\hat{\ell} = \mathop{\text{argmax}}_{\ell} \ \Pr(\ell|\mathbf{X}_k)$.

\vspace{-2mm}

\subsection{Classification  Metrics}
\label{subsec: classification_metrics}
For the tractable case of  linear classification, relevant metrics are characterized as follows. 
\subsubsection{Inference Uncertainty}
Inference uncertainty is measured using the  \emph{entropy of posteriors} commonly used  for ML and DNN classifiers \cite{Liu2020TWC,Liu2020TCCN}. Given the feature distribution in~\eqref{eq: sysmodel_signal_pdf}, the metric, denoted as  $H\left(\mathbf{x}_{k}\right)$,  for linear classification of  the partial feature vector $\mathbf{x}_{k}$ is given as
\begin{align}
H(\mathbf{x}_{k}) 
&\triangleq -\sum_{\ell=1}^{L} \Pr( \ell | \mathbf{x}_{k},\mathcal{W}_k) \log \Pr( \ell | \mathbf{x}_{k},\mathcal{W}_k)=  \log\left(\sum_{\ell=1}^{L}e^{-z_{\ell}(k)}\right) + \frac{\sum_{\ell=1}^{L}z_{\ell}(k)e^{-z_{\ell}(k)}}{\sum_{\ell=1}^{L}e^{-z_{\ell}(k)}}. \label{eqn: linear_entropy}
\end{align}

\subsubsection{Discriminant Gain}
\label{subsubsec: feature_importance}
Pairwise discriminant gain measures the discernibility between two classes in a subspace of the feature space. We adopt the well known  metric of  \emph{symmetric Kullback-Leibler (KL) divergence} defined as follows \cite{Saon00NeurIPS}. To this end, consider a class pair, say class $\ell$ and $\ell^{\prime}$, and a feature subspace spanning  the dimensions indexed by the  set ${\mathcal{S}} = \{n_1, n_2, ..., n_{|\mathcal{S}|}\} \subseteq \mathcal{W}$. Let  $G_{\ell,\ell^{\prime}}({\mathcal{S}})$ denote the corresponding  pairwise  discriminant gain. Based on the discussed data distribution model, the feature vector in this subspace is a uniform mixture of $L$ Gaussian components with mean $\tilde{\mu_{\ell}}$ and covariance matrix $\tilde{\mathbf C}$ ($\ell=1,2,...,L$), where $\tilde{\mu}_{\ell}(i)=\mu_{\ell}(n_i)$ and
\begin{equation}
\tilde{C}_{i,j} = 
\begin{cases}
C_{n_i,n_i}, & i=j,\\
0, &\text{otherwise}.
\end{cases}
\end{equation}
Based on the metric of  symmetric KL divergence, pairwise discriminant gain is evaluated as, 
\begin{align}
G_{\ell,\ell^{\prime}}\left({\mathcal{S}}\right) 
&\triangleq { {\sf{KL}}\left(\mathcal{N}(\tilde{\mu}_\ell, \tilde{\mathbf{C}})\parallel\mathcal{N}(\tilde{\mu}_{\ell^{\prime}}, \tilde{\mathbf{C}})\right)  + {\sf{KL}}\left(\mathcal{N}(\tilde{\mu}_{\ell^{\prime}}, \tilde{\mathbf{C}})\parallel\mathcal{N}(\tilde{\mu}_{\ell}, \tilde{\mathbf{C}})\right)} \nonumber\\
&= {\left(\tilde{\mu}_{\ell}-\tilde{\mu}_{\ell^{\prime}}\right)^{T}\tilde{\mathbf{C}}^{-1}\left(\tilde{\mu}_{\ell}-\tilde{\mu}_{\ell^{\prime}}\right)} 
= \sum_{n \in \mathcal{S}}g_{\ell,{\ell^{\prime}}}\left(n\right),
\label{eqn: linear_pair_dgain}
\end{align}
where  $g_{\ell,{\ell^{\prime}}}\left(n\right)=\frac{\left[\mu_\ell(n)-\mu_{\ell^{\prime}}(n)\right]^2}{C_{n,n}}$. The term $g_{\ell,{\ell^{\prime}}}\left(n\right)$ is named the pairwise discriminant gain in dimension-$n$. We define the average discriminant gain over  all $\frac{L(L-1)}{2}$ class pairs as 
\begin{eqnarray}
\bar{G}({\mathcal{S}}) = \frac{2}{L(L-1)}\sum_{  \ell<{\ell^{\prime}}\leq L}G_{\ell,\ell^{\prime}}\left({\mathcal{S}}\right)=  \sum_{n \in \mathcal{S}}\bar{g}\left(n\right), \label{eqn: d_gain_additivity}
\end{eqnarray}
where $\bar{g}(n)=\frac{2}{L(L-1)}\sum_{ \ell<{\ell^{\prime}}\leq L }g_{\ell,{\ell^{\prime}}}\left(n\right)$ is the average discriminant gain of dimension-$n$. We refer to \eqref{eqn: d_gain_additivity} as  the \emph{additivity} of discriminant gains. For binary classification, the subscripts of  $g_{1, 2}\left(n\right)$ and $G_{1, 2}\left({\mathcal{S}}\right) $ are omitted for simplicity, yielding $g\left(n\right)$ and $G\left({\mathcal{S}}\right)$.

\section{ProgressFTX Protocol and Control Problem}

\subsection{ProgressFTX Protocol}
\label{subsec: protocol}

\begin{figure}[t!]
\centering
\includegraphics[scale=0.22]{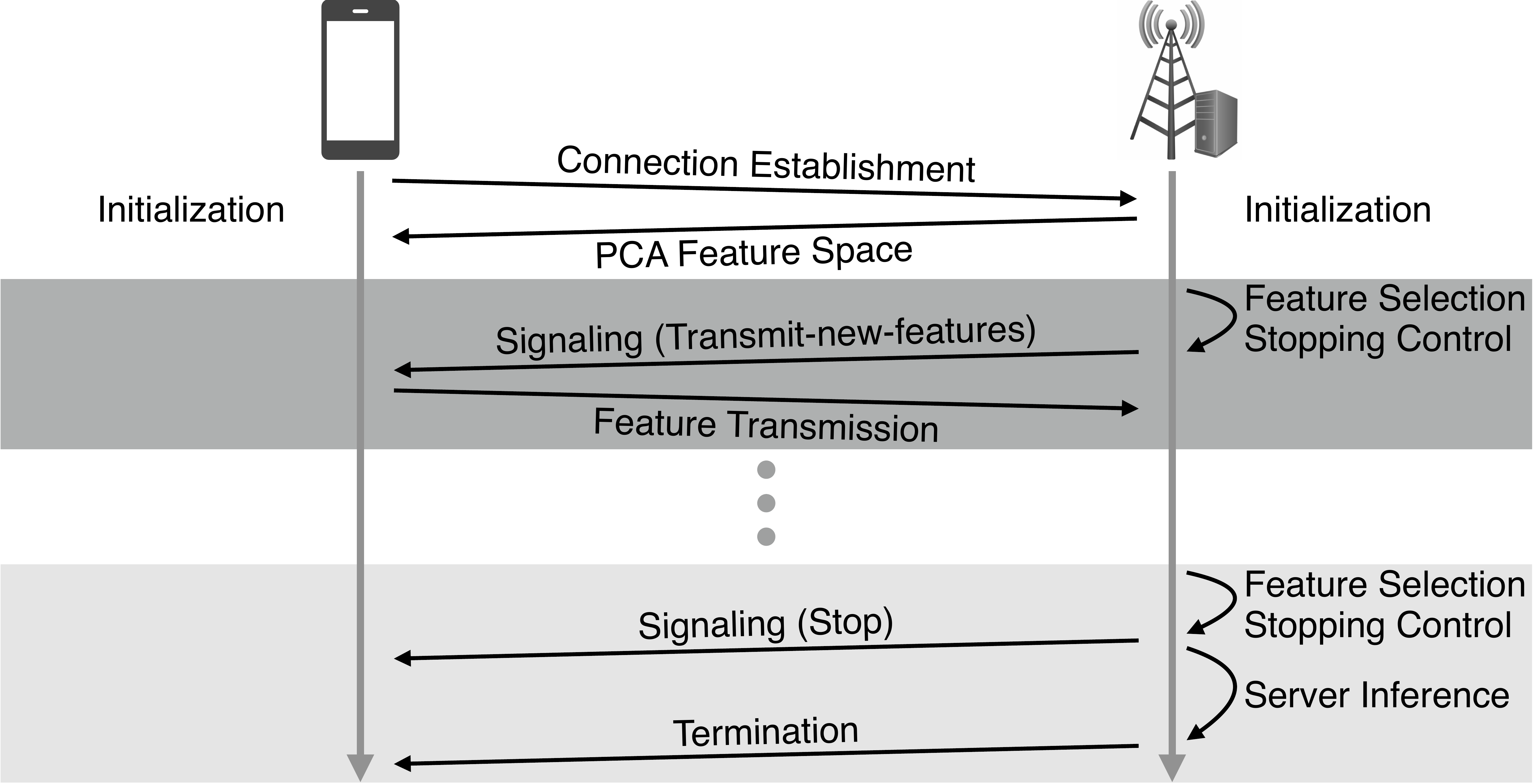}
\vspace{-5mm}
\caption{Illustration of the ProgressFTX protocol.  }
\label{fig: protocol}\vspace{-5mm}
\end{figure}

The ProgressFTX protocol implements a time window of continuous progressive feature transmission of a single data sample. If the remote inference fails to achieve the target accuracy, the task is declared a failure for a latency sensitive application (e.g., autonomous driving) or otherwise another instance of progressive transmission may be attempted after some random delay. The steps in the ProgressFTX protocol are illustrated in Fig.~\ref{fig: protocol} and described as follows. 
\begin{enumerate}
	\item \emph{Feature selection:} At the beginning of a slot intended for feature transmission, the server selects the feature dimensions and informs the device to transmit  corresponding features. For the case of linear classification, feature selection uses the importance metric of discriminant gain (see Section~\ref{subsec: classification_metrics}); for the other case of CNN classification,  the selection of feature maps depend on their importance levels (i.e., descent gradients as elaborated in Section~\ref{subsec: feature_selection_cnn}) which are available from the training phase and to be used in the ensuing phase of edge inference. Two facts should be emphasized: 1)  both methods are data-sample independent, which facilitates efficient implementation at the server  since it has no access to samples; and 2) feature importance is evaluated using a trained model, but not during the training. Specifically, the server records received features in $\mathbf{x}_{k}$ and their indices in $\mathcal{W}_k \subseteq \mathcal{W}$. The indices of selected features are stored in $\mathcal{S}_k$, where $\mathcal{S}_k \subseteq \mathcal{W}\setminus\mathcal{W}_k $.  A subset $\mathcal{S}_k$ satisfying this constraint is termed an  \emph{admissible subset}. Moreover, the number of features to be transmitted, $|\mathcal{S}_{k}|$, is limited by the number of untransmitted features  and the transmission rate: $\vert \mathcal{S}_{k} \vert=\min\{N- \vert\mathcal{W}_{k} \vert,Y_0\}$.

	\item \emph{Stopping control:} The number of features of a particular sample needed for remote classification to meet the accuracy requirement is sample dependent (for both linear and CNN classifications). Neither the device nor the server has prior knowledge of this number since each has access to either the sample or the model but not both. Online stopping control aims at minimizing the number of transmitted features under the said requirement. Its procedure is described as follows while its optimization problem is formulated in Section~\ref{sec: problem_formulation} and solved in Section~\ref{sec: linear_feature_stopping}.  Let $b_k$ indicate the server's decision on whether the device should transmit features in  slot $k$ (i.e., $b_k = 1$), or \emph{stop} the progressive transmission (i.e., $b_k = 0$). {The decision is made using a derived  policy for the case of linear classification (see Section~\ref{sec: linear_feature_stopping}) or a regression model for uncertainty prediction  (see Section~\ref{subsec: cnn_stopping}).} If the decision  is  to transmit, proceed to the next step; otherwise, go to Step  $6$. 
	
	\item \emph{Feedback:} The decisions on feature selection and stopping control are communicated to the device by feedback of one of three available  control signals described as follows: \emph{(a)} If the channel is not in outage in the preceding slot, the feedback contains a \emph{transmit-new-features signal} and indices of features in  $\mathcal{S}_k$ that instructs  the device to transmit the features selected by the server in the current slot. \emph{(b)} Otherwise, the feedback is a \emph{retransmission signal} instructing  the device to retransmit the features transmitted in the preceding  slot. In this case, a  link-layer HARQ protocol is employed to enable retransmission. \emph{(c)} If the server decides to stop the transmission (i.e., $b_{k}=0$), a \emph{stop signal} is fed back to instruct the device to terminate feature transmission.

	\item \emph{Feature transmission:} Upon receiving a transmit-new-features signal with indices $\mathcal{S}_k$,  the device transmits  the \emph{incremental} feature vector,  $\Delta\mathbf{x}_{k}$, that comprises the selected features:
	\begin{equation}
	\Delta\mathbf{x}_{k}=\left[{x}(n_1), {x}(n_2), ..., {x}(n_{|\mathcal{S}_k|})\right]^{T}, 
	\label{eqn: delta_x_k}
	\end{equation}
	where $n_1,...,n_{|\mathcal{S}_k|}$ are indices in $\mathcal{S}_k$.
	At the end of slot-$k$, the server updates the set of received features $\mathcal{W}_{k+1} = \mathcal{W}_{k} \bigcup \mathcal{S}_k$ and 
	the partial feature vector $\mathbf{x}_{k+1}  = \mathbf{x}_k + \mathbf{M}\Delta\mathbf{x}_{k}$, where the entry in the $i$-th row and $j$-th column of the placement matrix $\mathbf{M}$ is 
	\begin{equation}
	\label{eqn: update_mapping}
	M_{i,j} = 
	\begin{cases}
	1, & j \leq |\mathcal{S}_{k}| \text{ and } i = n_j,\\
	0, & \text{otherwise.}
	\end{cases}
	\end{equation}
	
	\item \emph{Server inference:} The server infers an estimated label $\hat{\ell}$ and its uncertainty   level using  the partial feature vector $\mathbf{x}_k$ and a trained model as discussed in  Section~\ref{subsec: classification_model}. Then we go to Step 1 and restart ProgressFTX for another data sample.
	
	\item \emph{Termination:} Upon receiving a \emph{stop signal}, the device terminates the process of progressive transmission. The latest estimated label is downloaded to the device.  
	
\end{enumerate}

In summary, feature selection, stopping control, feature transmission and uncertainty evaluation are executed sequentially in each slot. The transmission is terminated at the time when the uncertainty is evaluated to fall below the target or uncertainty reduction is outweighed by the corresponding communication cost.

\subsection{Control Problem Formulation}
\label{sec: problem_formulation}

ProgressFTX has two objectives: 1)  maximize the uncertainty reduction (or equivalently, improvement in inference accuracy), and 2) minimize the communication cost. The optimal stochastic control is formulated as a dynamic programming problem as follows. We define the system state  observed by the server at the beginning of slot $k$ as the following tuple
\begin{equation}
\label{eqn: def_system_state}
\theta_k \triangleq \left(\mathbf{x}_{k},\mathcal{W}_k\right).
\end{equation}
Recall that $Y_0$ represents the transmission rate in slot $k$, $\mathbf{x}_{k}$ the received  partial feature vector, and $\mathcal{W}_k$ the indices of received features. The control policy, denoted as $\Omega$,  maps  $\theta_k$ to the feature-selection \emph{action},  $\mathcal{S}_k$,  and stopping-control action, $b_k$, 
$\Omega: \theta_k \to \left(\mathcal{S}_k,b_k\right)$. The net reward of transitioning from  $\theta_{k}$ to $\theta_{k+1}$  is the decrease in  inference uncertainty  minus the communication cost, which can be written as 
\begin{equation}
u(\theta_k,\Omega)=H\left(\mathbf{x}_{k}\right) - H\left(\mathbf{x}_{k+1}\right) -b_k c_0,
\end{equation}
where $c_0$ denotes the transmission cost of one slot.
For sanity check, $u(\theta_k,\Omega)=0$ if  $b_k=0$. The net reward for slot $k$ requires server's evaluation based on transmitted features if the control decision is {proceeding transmission or retransmisson;  and thus,} is unknown before making the decision for the slot. Hence,  the expected net reward should be considered as the utility function. {Without loss of generality, consider} $K$ slots with the current slot set as slot $1$. Then the  system utility is defined as the  expectation of the sum rewards over $K$ slots conditioned on the system state and control policy:
\begin{equation}
\label{eqn: form_objective}
U(\theta_1,\Omega) \triangleq \mathbb{E}_{\{\theta_k\}_{k=2}^{K}}  \left[\sum_{k=1}^{K}u(\theta_k, \Omega)\ \Bigg\vert\ \theta_1, \Omega \right]. 
\end{equation}
The ProgressFTX control problem can be readily  formulated  as the following dynamic program: 
\begin{equation*}\text{(P1)}\quad
\begin{aligned}
\max\limits_{\Omega}\quad& U(\theta_1,\Omega) \\
\mathrm{s.t. }\quad& b_{k} \in \{0,1\},  \ k = 1, \ 2, \ ..., \ K, \\
& b_{k+1} \leq b_{k},  \ k = 1, \ 2, \ ..., \ K-1, \\
&    \vert \mathcal{S}_{k} \vert=\min\{N- \vert\mathcal{W}_{k} \vert,Y_0\}, \ k = 1, \ 2, \ ..., \ K, \\
& \mathcal{S}_k \subseteq \left(\mathcal{W}\setminus\mathcal{W}_k \right), \ k = 1, \ 2, \ ..., \ K .
\end{aligned}
\end{equation*}
The complexity of solving the problem using a conventional iterative algorithm (e.g., value iteration) is prohibitive due to the curse of dimensionality as the state space in~\eqref{eqn: form_objective} is not only high dimensional but also partially continuous. The difficulty is overcome in the sequel via analyzing the structure of the optimal policy.

\section{ProgressFTX Control -- Optimal Feature Selection}
\label{sec: linear_feature_selection}
In this section, targeting  linear classification, the optimal policy for feature selection in the ProgressFTX protocol is derived in closed form by {first approximating the objective function in Problem (P1),} and then solving the problem. 

\subsection{Approximation  of Expected Uncertainty}
\subsubsection{Binary Classification}
For ease of exposition, we first consider the simplest case of binary classification (i.e., $L = 2$) before extending the results to the  general case. {The goal of the subsequent analysis is to approximate the expected uncertainty function of unknown features to be transmitted in the following $K$ slots, given the features already received by the server, which is the objective of Problem (P1) in \eqref{eqn: form_objective}.} To begin with,  
the objective can be rewritten as 
\begin{equation}
U(\theta_1,\Omega) = 
\mathbb{E}_{\{\Delta\mathbf{x}_{k}\}_{k=1}^{K}}\left[ \sum_{k=1}^{K} H\left(\mathbf{x}_{k}\right) - H\left(\mathbf{x}_{k+1}\right)\ \Bigg\vert\ {\theta}_{1}, \Omega \right] - c_0\sum_{k=1}^{K}b_k.\nonumber 
\end{equation}
Based on the fact that feature distributions in different dimensions are independent {(as the covariance matrix is diagonal)}, telescoping gives 
\begin{equation}
U(\theta_1,\Omega) = H\left(\mathbf{x}_{1}\right) - \mathbb{E}_{\{\Delta\mathbf{x}_{k}\}_{k=1}^{K}}\left[ H\left(\mathbf{x}_{K+1}\right)\ \Big\vert\ {\theta}_{1},\Omega \right] - c_0\sum_{k=1}^{K}b_k .
\label{eqn: simplified_objective}
\end{equation}
In the remainder of the sub-section, we show that the expected uncertainty term in  $U(\theta_1,\Omega)$ can be reasonably approximated by  a monotone decreasing function of the average discriminant gain of selected features, which facilitates the subsequent feature-selection optimization.

Consider first the feature-selection decisions for transmission over $K$ slots specified by $\mathcal{S} =  \bigcup_{k=1}^{K} \mathcal{S}_{k}$. {Upon receiving the selected features $\Delta\mathbf{x}$, the partial feature vector  $\mathbf{x}_1$ is updated as $\mathbf{x}_{K+1} = \mathbf{x}_1+\mathbf{M}^{(K)}{\Delta \mathbf{x}}$, where the placement matrix $\mathbf{M}^{(K)}$,   defined similarly as $\mathbf{M}$ in~\eqref{eqn: update_mapping}, places all the features transmitted over $K$ slots as recorded in $\mathcal{S}$.} For an arbitrary feature vector $\mathbf{x}$,  define the differential-distance as  $\delta(\mathbf{x})=z_1(\mathbf{x})-z_2(\mathbf{x})$, where $z_1$ and $z_2$ represent the half squared Mahalanobis distances from $\mathbf{x}$ to the centroids of  classes $1$ and $2$, respectively. For ease of notation, let $\delta(\mathbf{x}_1)$ and $\delta(\mathbf{x}_{K+1})$ be denoted as $\delta_1$ and $\delta_{K+1}$, respectively. In particular, 
\begin{equation}
\delta_{K+1} = \frac{1}{2}\sum_{n\in(\mathcal{W}_1\bigcup\mathcal{S})} \frac{  [\mu_2(n)-\mu_1(n)][2x(n)-\mu_2(n)-\mu_1(n)]}{C_{n,n}}.
\label{eqn: delta-prime}
\end{equation}
Note that $\delta_{K+1}$ is a random variable from the server's perspective before receiving the selected features. Given the GM model~\eqref{eq: sysmodel_signal_pdf}, 
the distribution of $\delta_{K+1}$ is characterized next. 
\begin{Lemma}
\label{lemma: delta_distribution}\emph{
Given the partial feature vector  $\mathbf{x}_{1}$  and feature-selection decision $\mathcal{S}$, the differential distance $\delta_{K+1}$  is a mixture of two Gaussian components with the following conditional PDF:}
\begin{equation}
\label{eqn: analysis_binary_pdf_delta}
f_{\delta_{K+1}}\left(\delta  |  \mathbf{x}_1,\mathcal{S}\right) = \frac{1}{2}\left[\mathcal{N}\left(\delta \ \bigg|\ \delta_1+\frac{G(\mathcal{S})}{2},G(\mathcal{S})\right)+\mathcal{N}\left(\delta \ \bigg|\ \delta_1-\frac{G(\mathcal{S})}{2},G(\mathcal{S})\right)\right], 
\end{equation}
\emph{where $G(\mathcal{S})=G_{1,2}(\mathcal{S})$ is the total discriminant gain of features selected in $\mathcal{S}$ given in~\eqref{eqn: linear_pair_dgain}.} 
\end{Lemma}

Given a differential distance $\delta$,  the inference uncertainty $H\left(\mathbf{x}\right)$ in \eqref{eqn: linear_entropy} can be simplified as: 
\begin{equation}
H(\delta)
= \log\left(1+e^{-\delta}\right) + \frac{\delta}{e^{\delta}+1}.
\end{equation}
Then the expected  uncertainty resulting from the feature-selection decision $\mathcal{S}$  can be written as
\begin{eqnarray}
\mathbb{E}\left[H\left(\delta_{K+1}\right)  \Big\vert \ \theta_1,\mathcal{S} \right]=\int_{-\infty}^{\infty} H(\delta)f_{\delta_{K+1}}\left(\delta  | \mathbf{x}_1,\mathcal{S}\right) {\rm d} \delta \triangleq \bar{H}(\delta_1,G(\mathcal{S})) .
\label{eqn: binary_analysis_int_delta}
\end{eqnarray}
In view of the difficulty in deriving the above integral in closed form, we resort to its approximation using an upper bound as shown in the following lemmas. 
\begin{Lemma}
\label{lemma: ub_infer_uncertainty}
\emph{Given $\delta \in \mathbb{R}$, the inference uncertainty $H(\delta)$ can be upper bounded by
\begin{equation}
 \label{eqn: h_ub_definition}
 H^{\sf ub}(\delta) = (1+|\delta|)e^{-|\delta|}.
\end{equation}} 
\end{Lemma}
\begin{proof}
The result follows from the fact that the ratio $\frac{H(\delta)}{H^{\sf ub}(\delta)}$ has the maximum value of $\log{2}$, which is smaller than one,  at  $\delta = 0$.
\end{proof}
For large $\delta$, the bound is tight since $\lim_{|\delta|\to\infty}\frac{H(\delta)}{H^{\sf ub}(\delta)}=1$.
Using Lemma~\ref{lemma: ub_infer_uncertainty}, we can bound the expected uncertainty as $\bar{H}(\delta_1,G(\mathcal{S}))\leq \bar{H}^{\sf ub}(\delta_1,G(\mathcal{S}))$ with 
\begin{equation}
 \bar{H}^{\sf ub}(\delta_1,G(\mathcal{S})) = \int_{-\infty}^{\infty}H^{\sf ub}(\delta)f_{\delta_{K+1}}\left(\delta | \mathbf{x}_1,\mathcal{S}\right) {\rm d} \delta,  
\label{eqn: ub_tight_binary}
\end{equation}
where  $H^{\sf ub}(\delta)$ is given in Lemma~\ref{lemma: ub_infer_uncertainty}. The above analysis allows approximating $\bar{H}(\delta_1,G(\mathcal{S}))$ by a simple exponential upper bound as shown in Lemma \ref{lemma: exponential_bound}, proven in Appendix~A. This  approximation is also  {validated} numerically in Fig.~\ref{fig: expected_inference_uncertainties_example}.
\begin{Lemma}
\label{lemma: exponential_bound}
\emph{The expected uncertainty can be  upper bounded by an exponential function as
\begin{equation}
   \label{eqn: tilde_H}
   \bar{H}(\delta_1,G(\mathcal{S}))  \leq \bar{H}^{\sf ub}(\delta_1,G(\mathcal{S})) \leq \tilde{H}(\delta_1,G(\mathcal{S})) \triangleq c_1 e^{-c_2 G(\mathcal{S})},
\end{equation}
where $c_1$ and $c_2$ are positive constants depending on $\delta_1$. } 
\end{Lemma}

\begin{figure}[t]
\centering
\subfigure[$\delta_1=0$]{\includegraphics[height=6cm]{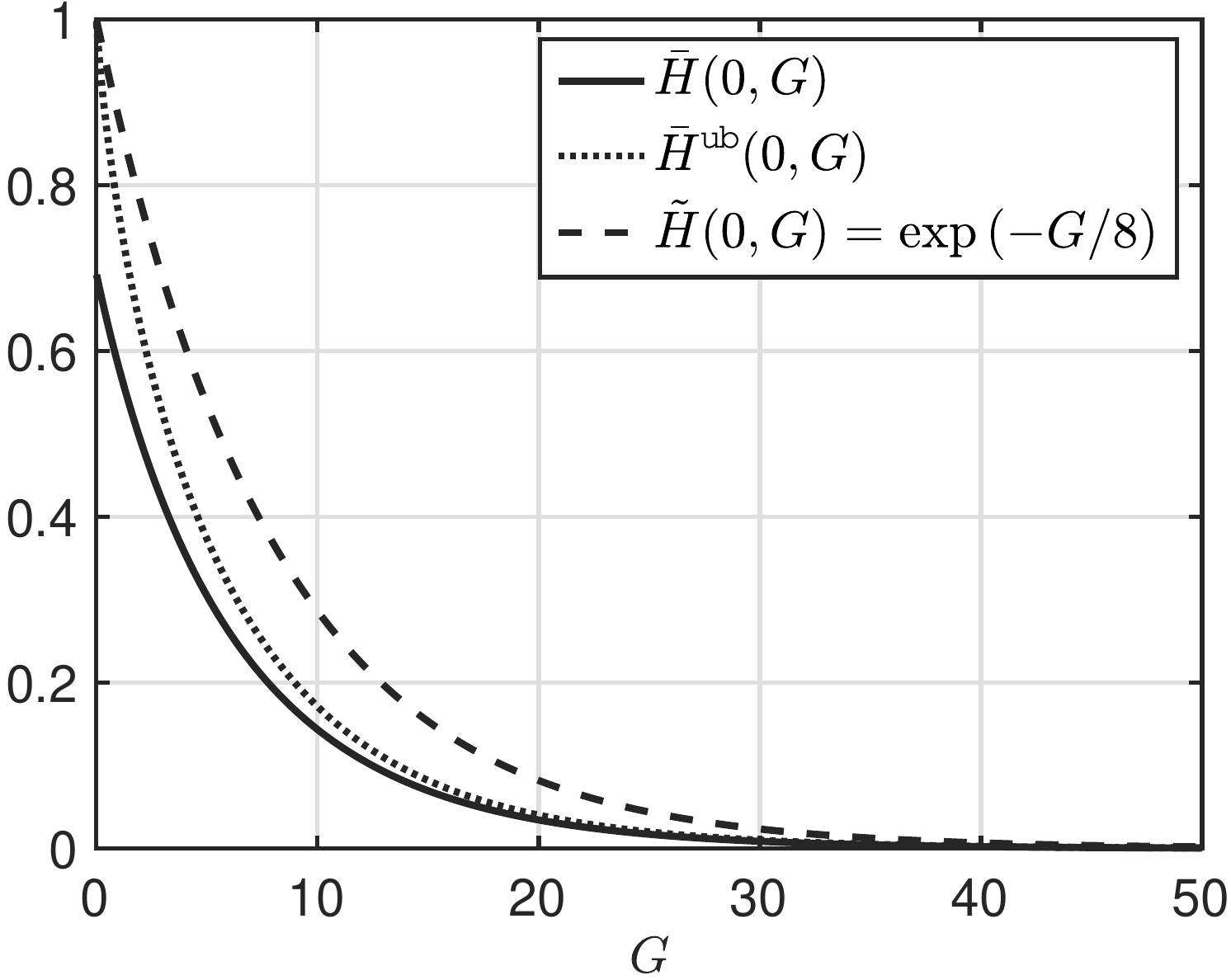}}
\hspace{1cm}
\subfigure[$\delta_1=1.2$]{\includegraphics[height=6cm]{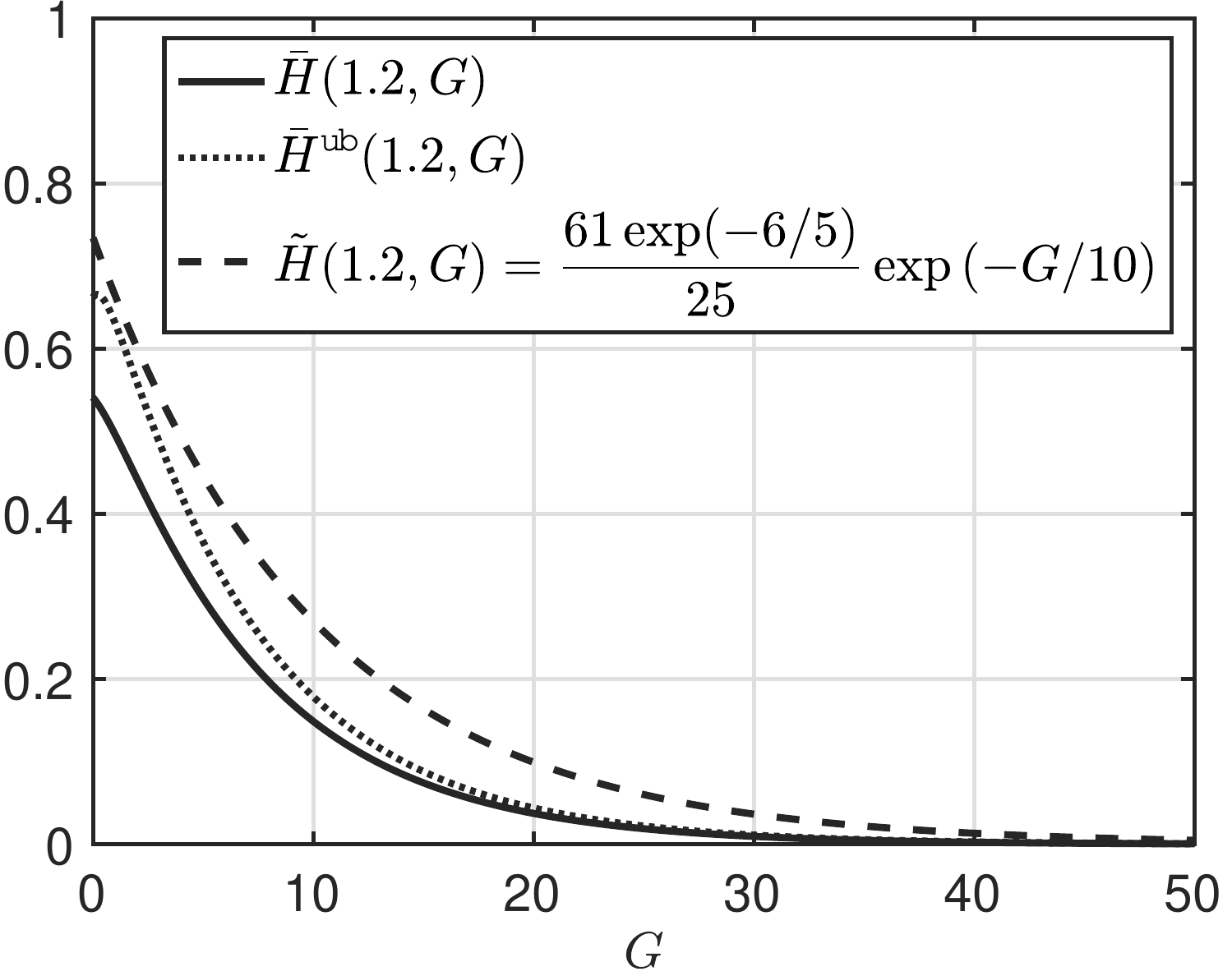}}\vspace{-2mm}
\caption{Comparisons between the expected inference uncertainty and its  upper bounds in~\eqref{eqn: tilde_H} for (a) $\delta_1=0$ or (b) $\delta_1=1.2$.}
\label{fig: expected_inference_uncertainties_example}\vspace{-4mm}
\end{figure}

\subsubsection{Extension to Multi-class Classification}
The above analysis for binary classification can be generalized for  multi-class classification (i.e., $L\geq 3$) as follows. The inference uncertainty in the case of multi-class classification can be upper bounded in terms of pairwise uncertainties,
\begin{align}
H\left(\mathbf{x}_{k}\right)&\leq \sum_{\ell=1}^{L-1}\sum_{\ell^{\prime}
=\ell+1}^{L} \! \log\left(e^{-z_{\ell}(k)}+e^{-z_{\ell^{\prime}}(k)}\right) \!+\! \frac{z_\ell(k)e^{-z_{\ell}(k)} \!+\! z_{\ell^{\prime}}(k)e^{-z_{\ell^{\prime}}(k)}}{e^{-z_{\ell}(k)} \!+\! e^{-z_{\ell^{\prime}}(k)}} = \sum_{\ell=1}^{L-1}\sum_{\ell^{\prime}=\ell+1}^{L} \! H\left(\delta^{(\ell,\ell^{\prime})}_k\right), \nonumber
\end{align}
where $\delta^{(\ell,\ell^{\prime})}_k\triangleq z_\ell(k)-z_{\ell^{\prime}}(k)$ and the equality is approached in the asymptotic case that there exists an infinitely small $\delta^{(\ell,\ell^{\prime})}_k$ compared with other pairwise differential distances. Moreover, the distribution of $\delta^{(\ell,\ell^{\prime})}_k$ conditioned on a partial feature vector $\mathbf{x}_1$ and a feature-selection decision $\mathcal{S}$ is a mixture of $L$ scalar Gaussian components, similar to its counterpart in binary classification~\eqref{eqn: analysis_binary_pdf_delta}. As a result, an upper bound on the conditional expected inference uncertainty can be obtained following the approach for binary classification, i.e., substituting $H\left(\delta^{(\ell,\ell^{\prime})}_k\right)$ with $H^{\sf ub}\left(\delta^{(\ell,\ell^{\prime})}_k\right)$ as that in~\eqref{eqn: ub_tight_binary}. Given the above straightforward procedure, the remainder of the paper assumes binary classification to simplify  notation and exposition.

\subsection{Optimal Feature Selection}
\label{subsec: optimal_feature_selection}
First, Problem (P1) reduces  to the problem of optimal feature selection by conditioning on $k^*$  subsequent transmissions, i.e., transmission is stopped at slot $(k^*+1)$. The decision can be captured by defining a sequence of virtual transmission rates $\{\tilde{Y}_k\}$ such that  $\tilde{Y}_k = Y_0$ for $k \leq k^*$ and $g_{k}\geq g_0$, while $\tilde{Y}_k = 0$ otherwise. Moreover, approximate the function  $\bar{H}(\delta_1, G(\mathcal{S}))$ in the objective in Problem (P1) by its upper bound $\tilde{H}(\delta_1,G(\mathcal{S}))$ in Lemma~\ref{lemma: exponential_bound}. Given the stopping decision and approximate objective, Problem (P1) reduces to the following problem. 
\begin{equation*}\text{(P2)} \quad
\begin{aligned}
\min\limits_{\{\mathcal{S}_k\}_{k=1}^{K}}\quad& \tilde{H}(\delta_1, G(\mathcal{S}))  \\
\mathrm{s.t. }\quad& \vert \mathcal{S}_{k} \vert=\min\{N- \vert\mathcal{W}_{k} \vert,\tilde{Y}_k\}, \ k = 1, \ 2, \ ..., \ K, \\
                 & \mathcal{S}_k \subseteq \left(\mathcal{W}\setminus\mathcal{W}_k \right), \ k = 1, \ 2, \ ..., \ K, \\
                 & \mathcal{S} = \bigcup_{k=1}^{K} \mathcal{S}_{k}.
\end{aligned}
\end{equation*}
It follows from the monotonicity of the objective that  the optimal policy for feature selection in each slot is to choose from not yet received features whose indices correspond to the maximum discriminant gains. Let the importance of each feature be associated with its discriminant gain.  Then the optimal policy, termed \emph{importance-aware feature selection}, is as  presented in Algorithm~\ref{algorithm: feature_selection}.

\begin{algorithm}[t]
\caption{Importance-aware Feature Selection at the Server}
\label{algorithm: feature_selection}
\textbf{Input:} Set of all feature indices $\mathcal{W}$, set of received feature indices $\mathcal{W}_1$, and virtual transmission rates $\{\tilde{Y}_{k}\}_{k=1}^{k*}$\;
\textbf{Initialize:} Optimal feature subsets $\mathcal{S}_k^{\star}=\emptyset, \forall k = 1,2,...,k^{*}$\;
\textbf{for} $k=1,2,\cdots,k^*$ \textbf{do}\\
    \begin{enumerate}
        \item[1:] \emph{Feature selection for one slot.} \\
              \textbf{for} $i=1,2,\cdots,\min\{N-\vert\mathcal{W}_k\vert, \tilde{Y}_k \}$ \textbf{do}\\
                            \begin{itemize}
                            \item[] {Select the most important feature from the admissible set $\left(\mathcal{W}\setminus\mathcal{W}_k\setminus\mathcal{S}_{k}^{\star}\right)$, and add the selected feature index to $\mathcal{S}_{k}^{\star}$, i.e., $\mathcal{S}_{k}^{\star} := \mathcal{S}_k^{\star}\bigcup \Bigg\{{\mathop{\text{argmax}}\limits_{n \in \left(\mathcal{W}\setminus\mathcal{W}_k\setminus\mathcal{S}_{k}^{\star}\right) }\! g(n)} \Bigg\}$;}\\
                            \end{itemize} 
            \textbf{end for}
        \item[2:] \emph{Update.} Update the set of virtually received features' indices as $\mathcal{W}_{k+1}=\mathcal{W}_{k}\bigcup \mathcal{S}_{k}^{\star}$\;
    \end{enumerate}
\textbf{end for} \\    
\textbf{Output:} $\mathcal{S}_1^{\star},  \mathcal{S}_2^{\star}, ...,  \mathcal{S}_{k^{*}}^{\star}$.
\end{algorithm}

\section{ProgressFTX Control -- Optimal Stopping} 
\label{sec: linear_feature_stopping}

{Given the importance-aware feature selection process in the preceding section, we describe the remaining design of  ProgressFTX to solve the problem of  optimal stopping control}. Following the notation used in Section~\ref{subsec: optimal_feature_selection}, transmission is assumed to stop in slot $(k^{*}+1)$, resulting in  $k^{*}$ transmitted incremental feature vectors.
The  stopping-control indicators $\{b_k\}$ and $k^{*}$ can be related as 
\begin{equation}
b_k = 
\begin{cases}
1, & \forall 1\leq k \leq k^{*},\\
0,  & \forall k^{*}+1 \leq k \leq K.
\end{cases}
\label{eqn: relate_b_D}
\end{equation} 
As before, for tractability, we approximate  the objective in Problem (P1) (see~\eqref{eqn: simplified_objective}) using the upper bound $\tilde{H}(\delta_1,G(\mathcal{S}))$ in Lemma~\ref{lemma: exponential_bound}. Then given the  optimally selected and transmitted  features~$\{\mathcal{S}^{\star}_k\}$, Problem (P1) can be approximated as 
\begin{equation*}\text{(P3)} \quad
\begin{aligned}
\min\limits_{k^{*}}\quad& \mathbb{E}_{\{g_k\}_{k=1}^{K}}\left[\tilde{H}(\delta_1,G(\mathcal{S^{\star}})) \ \big\vert \ \theta_1 \right] + c_0 k^{*}  \\
\mathrm{s.t. }\quad& k^{*} \in \{0,1,...,K\},\\
                 & \mathcal{S}^{\star} = \bigcup_{\overset{1\leq k \leq k^{*}}{g_k\geq g_0}} \mathcal{S}_{k}^{\star}.
\end{aligned}
\end{equation*}
Note the dependence of the equivalent optimal feature selection $\mathcal{S}^{\star}$ on the random channel gain  $g_k$. Problem (P3) belongs to the class of optimal-stopping problems \cite{Chow1971GreatExpectations}. In operation, the server makes a binary decision (i.e., $b_1 = 0$ or $1$) on whether to transmit in the current slot. The solution for Problem (P3) can be converted into the control decision, $b_1$, by  
\begin{equation}
\label{eqn: relate_D_b1}
b_1=\min\{1,k^{*}\}.
\end{equation}
The optimal control policy (also known as optimal stopping-rules) are derived  for both the Gaussian and fading channels in the following subsections.

\subsection{Optimal Stopping with Gaussian Channels}
In this subsection, we address the optimal stopping for Gaussian channels, where $g_k=g_0,  \forall k$. In this case, the optimal feature subsets $\{\mathcal{S}_k^{\star}\}$ for all current and future slots can be determined by the importance aware feature selection (i.e., Algorithm~\ref{algorithm: feature_selection}) at the beginning of the current slot. Given the total throughput in $k^{*}$ proceeding transmissions, $Y_0 k^{*}$, the total number of features selected for transmission is $\vert\bigcup_{k=1}^{k^{*}}\mathcal{S}_{k}^{\star}\vert=\min\{N-\vert\mathcal{W}_1\vert, Y_0 k^{*}\}$ where $N-\vert\mathcal{W}_1\vert$ specifies the number of untransmitted features in the current slot. Let $\mathrm{G}^{\star}(\theta_1,k^{*})$ denote the discriminant gain of optimally selected features subject to a number of transmissions $k^{*}$ and conditioned on $\theta_1$. Then the cumulative  discriminant gain over  $k^{\star}$ transmissions is obtained as  $G(\mathcal{S}^{\star})= \mathrm{G}^{\star}(\theta_1,k^{*})$. 
Problem (P3) can be rewritten for a  Gaussian channel as
\begin{equation*}\text{(P4)}\quad
\begin{aligned}
\min\limits_{k^{*}}\quad&   \tilde{H}\left(\delta_1,\mathrm{G}^{\star}(\theta_1, k^{*})\right) + c_0 k^{*}   \\
\mathrm{s.t.}\quad& k^{*} \ \in \ \{0,  1,  ..., K \}. 
\end{aligned}
\end{equation*}
This problem is convex. To prove this, we first rewrite the cumulative discriminant gain as
\begin{eqnarray}
\label{eqn: def_optimal_gain_beta}
\mathrm{G}^{\star}(\theta_1, k^{*})
&=&   \mathop{\text{max}}\limits_{ \overset{\mathcal{S} \subseteq\left(\mathcal{W}\setminus\mathcal{W}_{k}\right)}{\ |\mathcal{S}|=\min\{ N-\vert\mathcal{W}_1\vert , Y_0 k^{*}\}} } \sum_{n\in \mathcal{S}}g(n) =\sum_{k=1}^{k^{*}}G\left(\mathcal{S}_{k}^{\star}\right).
\end{eqnarray}
Using the fact that  $G(\mathcal{S}_{k}^{\star})\geq G(\mathcal{S}_{k+1}^{\star})$ for all $k$, one can infer from the last equation  that 
\begin{equation}
\label{eqn: convex_G_k}
\mathrm{G}^{\star}(\theta_1,k^{*}-1) + \mathrm{G}^{\star}(\theta_1,k^{*}+1)  \geq 2\mathrm{G}^{\star}(\theta_1,k^{*}),
\end{equation}
and thereby conclude that $\mathrm{G}^{\star}(\theta_1,k^{*})$ is \emph{concave} \emph{with respect to} (w.r.t.) $k^{*}$.
Combining this fact and that $\tilde{H}(\delta_1, G)$ is a convex and monotone decreasing function of $G$ shows that Problem (P4) is also convex. Based on a standard approach in discrete convex optimization~\cite{Murota98DiscreteConvex}, the optimal solution for  Problem (P4) can be  derived as
\begin{equation}
\label{eqn: importance_static_beta_star}
k^{\star} = \min\{K,~\tilde{k}\},
\end{equation}
where $\tilde{k}$ is given by
\begin{equation}
\label{eqn: marginal_reward_criteria_gaussian}
\tilde{k} = \min\{k\in\{0,...,K-1\} \ | \ \tilde{H}\left(\delta_1,\mathrm{G}^{\star}(\theta_1, k)\right)-\tilde{H}\left(\delta_1,\mathrm{G}^{\star}(\theta_1,k+1)\right)\leq c_0 \}.
\end{equation} 
Combining \eqref{eqn: relate_D_b1}, \eqref{eqn: importance_static_beta_star}, and \eqref{eqn: marginal_reward_criteria_gaussian} leads to the following main result of the section.

\begin{Theorem}[\emph{Optimal Stopping for Gaussian Channels}]
\label{theorem: optimal_stopping_importance}
\emph{The stopping rule for Gaussian channels, which solves Problem (P4), is that the transmission in ProgressFTX should be terminated if $\tilde{H}\left(\delta_1,0\right)-\tilde{H}\left(\delta_1,\mathrm{G}^{\star}(\theta_1,1)\right)\leq c_0$. In other words, the optimal stopping policy is given as 
\begin{equation}
\label{eqn: stopping_rule_gaussian}
b_1^{\star}=
\begin{cases}
0,  & R \leq c_0, \\
1, &\text{otherwise,}
\end{cases}
\end{equation}
where the incremental  reward (i.e., reduction in classification uncertainty) from transmission in the current slot is denoted as $R =  \tilde{H}\left(\delta_1,0\right)-\tilde{H}\left(\delta_1,\mathrm{G}^{\star}(\theta_1,1)\right)$.} 
\end{Theorem}
The optimal stopping policy in Theorem~\ref{theorem: optimal_stopping_importance} is interpreted as follows. As $ k^{*}$ grows, the transmission cost $c_0 k^{*} $ accumulates at a constant speed $c_0$, while the incremental reward from transmission  in slot $k^{*}$ monotonically decreases due to the convexity of the function  $\tilde{H}\left(\delta_1,\mathrm{G}^{\star}(\theta_1,k^{*})\right)$ (see Lemma~\ref{lemma: exponential_bound} and~\eqref{eqn: convex_G_k}). For this reason, $k^{\star}$ is exactly the maximum  number of subsequent transmissions such that the incremental  reward from  transmitting in slot $(k^{\star}+1)$ is no larger than the per-slot transmission cost, yielding the stopping threshold. In the special case of $k^{\star}=0$, transmissions should be terminated in  the current slot. 

The  server  algorithm for controlling ProgressFTX, including both feature selection and optimal stopping, is summarized in Algorithm~\ref{algo: online_control}, where the index $k$ refers to the $k$-th slot in an edge inference task instead of that for slot $k$ from a current slot in control problems.

\begin{algorithm}[t]
\caption{ProgressFTX Control at the Server for Gaussian Channels}
\label{algo: online_control}
\textbf{Initialize:} Observe the system state of the first slot $\theta_1$\;
\textbf{repeat:}\\
    \begin{enumerate}
        \item[1:] \emph{Feature selection.} Determine optimal feature subset $\mathcal{S}^{\star}_{k}$ by Algorithm~\ref{algorithm: feature_selection};
        \item[2:] \emph{Stopping control.} Evaluate the optimal stopping control indicator $b_k^{\star}$ by~\eqref{eqn: stopping_rule_gaussian};
        \item[3:]  \textbf{if} $b_k^{\star}=1$ \textbf{then:}\\
        \begin{itemize}
            \item[1) ] \emph{Feature transmission.} Receive the incremental feature vector $\Delta\mathbf{x}_k$;
            \item[2) ] \emph{Partial feature vector update.} Update the partial feature vector by the rule $\mathbf{x}_{k+1}  = \mathbf{x}_k + \mathbf{M}\Delta\mathbf{x}_{k}$, with the placement matrix $\mathbf{M}$ given by~\eqref{eqn: update_mapping};
            \item[3)] \emph{System state evolution.} Evolve to state $\theta_{k+1}$\;
        \end{itemize}
        \textbf{end if}
    \end{enumerate}
\textbf{until} Scheduled stopping $b_k^{\star}=0$;\\
\emph{Server inference.} Infer the label $\hat{\ell}$ using~\eqref{eqn: linear_classifer}\;
\emph{Termination.} Feedback the inferred label $\hat{\ell}$ to device.
\end{algorithm}

\subsection{Optimal Stopping with Fading Channels}
In this sub-section, the optimal stopping policy for Gaussian channels is  extended to the case of  fading channels with outage. This requires that two issues be addressed. The first is the optimal retransmission strategy upon the occurrence of an outage event. Given  importance-aware feature selection, the features that fail to be received due to channel outage are most important ones among the undelivered features. This leads to  the following strategy. 
\begin{Lemma}
[\emph{Optimality of Feature Retransmission}]
\label{lemma: retransmission}
\emph{To minimize classification uncertainty, it is optimal to retransmit an incremental feature vector, whose transmission fails  due to channel outage, in the next slot.}
\end{Lemma}

The second issue is whether the optimal stopping problem in the current case retains the convexity of its Gaussian-channel counterpart. The answer is positive as shown in the sequel. Given the fixed communication rate $R_0$ and channel outage probability $p_{\sf o}$, the transmitted features can be successfully received with probability $1-p_{\sf o}$. {The number of successful transmissions {out of $k^{*}$ trials}, denoted by $k^{\prime}$, follows a Binomial distribution, $\text{Binom}\left(k^{*},1-p_{\sf o}\right)$.} The conditional \emph{probability mass function} (PMF) of $k^{\prime}$ is given by 
\begin{equation}
    \label{eqn: pmf_fading_success}
    p(k^{\prime}|k^{*}) = \binom{k^{*}}{k^{\prime}}(1-p_{\sf o})^{k^{\prime}}p_{\sf o}^{k^{*}-k^{\prime}}, \ k'=0,1, \cdots, k^{*}.
\end{equation}

Based on the retransmission strategy in  Lemma~\ref{lemma: retransmission}, the total discriminant gain of  features in $k^{\prime}$ successfully received incremental feature vectors is given by  $\mathrm{G}^{\star}(\theta_1,k^{\prime})$ defined in~\eqref{eqn: def_optimal_gain_beta}.  Combining $\mathrm{G}^{\star}(\theta_1,k^{\prime})$ and~\eqref{eqn: pmf_fading_success}, the expectation of  the upper bound of uncertainty after $k^{*}$ times of transmission is denoted as $\Phi(\theta_1,k^{*})$ and given by
\begin{eqnarray}
\Phi(\theta_1,k^{*}) \triangleq \mathbb{E}_{k^{\prime}}\left[\tilde{H}\left(\delta_1,\mathrm{G}^{\star}(\theta_1,k^{\prime})\right)  \big\vert  k^{*} \right] 
= \sum_{k^{\prime}=0}^{k^{*}}{\tilde{H}\left(\delta_1,\mathrm{G}^{\star}(\theta_1,k^{\prime})\right) p(k^{\prime}|k^{*}) }.
\end{eqnarray}
Then Problem (P3) can be rewritten for a fading channel with outage as \begin{equation*}\text{(P5)}\quad
\begin{aligned}
\min\limits_{k^{*}}\quad&  \Phi(\theta_1,k^{*}) + c_0 k^{*}   \\
\mathrm{s.t.}\quad& k^{*} \ \in \ \{0,  1,  ..., K \}. 
\end{aligned}
\end{equation*}
\begin{Lemma}
\label{lemma:Convexity:FadingChannel}
\emph{The problem of optimal stopping for a fading channel with outage, namely Problem (P5), is convex.}
\end{Lemma}
\begin{proof}
See Appendix~\ref{app: proof_convex_stopping_fading}.
\end{proof}
 Given the convexity and following the arguments used to solve Problem (P4), the optimal solution for Problem (P5) is derived to be 
\begin{equation}
    \label{eqn: optimal_sol_fading}
    k^{\star} = \min\{K,\tilde{k}\},
\end{equation}
where $\tilde{k}$ in this case is given by
\begin{equation}
\label{eqn: marginal_reward_criteria_fading}
\tilde{k} = \min\{k\in\{0,...,K-1\} \ | \ \Phi\left(\theta_1,0\right)-\Phi\left(\theta_1,k\right)\leq c_0 \}.
\end{equation} 
The optimal stopping rule follows  from~\eqref{eqn: relate_D_b1}, \eqref{eqn: optimal_sol_fading}, and \eqref{eqn: marginal_reward_criteria_fading}, 
leading to the following:

\begin{Theorem}[\emph{Optimal Stopping for Fading Channels}] \emph{The stopping rule for a fading channel with outage, which solves Problem (P5), is that the transmission in ProgressFTX should be terminated if  $\Phi(\theta_1,0)-\Phi(\theta_1,1)\leq c_0$. The corresponding optimal stopping policy is given as 
\begin{equation}
\label{eqn: stopping_rule_fading}
b_1^{\star}=
\begin{cases}
0,  & R \leq \frac{c_0}{1-p_{\sf o}}, \\
1, &\text{otherwise,}
\end{cases}
\end{equation}
where the incremental reward $R= \tilde{H}\left(\delta_1,0\right)-\tilde{H}\left(\delta_1,\mathrm{G}^{\star}(\theta_1,1)\right)$ follows that in Theorem~\ref{theorem: optimal_stopping_importance}.}
\label{theorem: optimal_stopping_importance_fading}
\end{Theorem}
Comparing Theorems~\ref{theorem: optimal_stopping_importance} and~\ref{theorem: optimal_stopping_importance_fading}, the optimal stopping policy in the case of fading channel retains the  threshold based structure of its Gaussian-channel counterpart. Nevertheless, the transmission threshold is higher in the former than that in the latter (by a factor of $\frac{1}{1-p_{\sf o}}$) due to the additional communication cost caused by channel outage. To be specific, given the same level of uncertainty reduction, the needed number of transmission, say   $k$, in the case of Gaussian channel is expected to increase to $\frac{k}{1-p_{\sf o}}$ when fading is present. As a result, a requirement on inference accuracy achievable in the Gaussian channel case may not be affordable in the fading channel case.

ProgressFTX control in Algorithm~\ref{algo: online_control} can be easily modified to account for outage in fading channels as follows. First,  the optimal stopping control indicator $b_k^{\star}$ shall be evaluated by~\eqref{eqn: stopping_rule_fading} instead. Next, Step 3 should be modified as follows. If the transmission in sub-step 1) is successful, then sub-step 2) is partial feature vector update as presented in Algorithm~\ref{algo: online_control}; otherwise,  sub-step 2) involves  the server sending  a retransmission signal to the device  for retransmission.

\section{ProgressFTX for CNN Classifiers}
\label{sec: cnn_model}

Given the complex architecture of a CNN model, ProgressFTX for a CNN classifier cannot be directly designed using an optimization approach as in the preceding sections for its linear counterpart. To overcome the difficulty, we leverage the principles underpinning the optimized ProgressFTX policies, namely importance-aware feature selection and optimal stopping,  for the linear model to design their counterparts under the CNN model as follows. 

\subsection{Importance-aware Feature Selection}
\label{subsec: feature_selection_cnn}

Consider the step of feature selection in the ProgressFTX protocol in Section~\ref{subsec: protocol}. The discriminant gain defined in~\eqref{eqn: linear_pair_dgain} for a linear classifier, which underpins the associated metric  of  feature importance,  is not applicable to a CNN model. In this case,  we propose the use of another suitable  metric proposed in \cite{Molchanov2019CVPR}, which is introduced next. Consider  the parameters $\{w_m\}$ of  the last CONV layer of the device sub-model. Let $\frac{\partial\mathcal{L}}{\partial w_m}$ denote the $m$-th entry of the gradient, i.e., the partial derivative of the learning loss function $\mathcal{L}$ w.r.t. the parameter $w_m$,  which is available from the last round of model training as computed using the back-propagation algorithm. Then its importance can be measured using the associated reduction 
of learning loss from retaining $w_m$,
approximated by the following first-order Taylor expansion of the squared loss of prediction errors~\cite{Molchanov2019CVPR}: $\tilde{I}(m)=\left(\frac{\partial\mathcal{L}}{\partial w_m} \cdot w_m\right)^2$. The importance of the $n$-th feature map is  defined to be the summed importance of parameters in the $n$-th filter outputting the map: $g_{\sf c}(n) = \sum_{w_m~\text{in the $n$-th filter} } {\tilde{I}(m)}$.
Since $\big\{ \frac{\partial\mathcal{L}}{\partial w_m} \big\}$ are readily available in training, the server is able to obtain the value of $\{g_{\sf c}(n)\}$ after training $f_{\sf d}(\cdot)$ and form a lookup table for reference during ProgressFTX control. Given $\{g_{\sf c}(n)\}$, the importance-aware feature selection for each slot is performed by selecting a given number of most important filters, whose feature maps will be transmitted in the slot. It should be reiterated that the number of selected feature maps is communication-rate dependent and may varies over slots. Algorithm~\ref{algorithm: feature_selection}  for the linear classifier can be modified accordingly. The details are straightforward, and thus omitted.

\subsection{Stopping Control Based on Uncertainty Prediction}
\label{subsec: cnn_stopping}

Consider the step of stopping control in the ProgressFTX protocol in Section~\ref{subsec: protocol}. The optimal stopping control in the current case is stymied by the difficulty in finding a tractable and accurate approximation of the expected classification uncertainty as a function of input features similarly to Lemma~\ref{lemma: exponential_bound} for linear classification. To tackle the challenge, we resort to an algorithmic approach in which a regression  model is trained to predict the uncertainty function of feature maps to be transmitted in the following $K$ slots given the feature maps already received by the server. One particular popular architecture of regression models in the literature is adopted \cite{SpechtTNN91,Gao2019NeuralComput}. To be able to predict the uncertainty for a future partial feature-map tensor, both the current partial feature-map tensor and the selected index subset are needed as the input. Since the underpinning prediction is based on the current partial feature maps, ProgressFTX control for CNN models is also sample-dependent. The architecture contains concatenation of two streams of regression features extracted from  the feature-map tensor, $\mathbf{X}_{k}$, and the index subset $\mathcal{S}_{k}$, into one intermediate feature map, fed into the deeper layers of the regression model (see Fig.~\ref{fig: regression_model}).
 
To train the regression model, a training dataset and a prediction loss function need to be properly designed.  The training dataset on the server, denoted as $\mathcal{D}$,  comprises labeled samples $\{\mathbf{D}_{i},H_{i}\}$, each with the a tuple of $\mathbf{D}_{i}=(\mathbf{X}_{(i)}, \mathcal{S}_{(i)})$ and a scalar label $H_{i}$, $i=1,2,...,|\mathcal{D}|$. Consider a tensor of all feature maps extracted by the server  sub-model  $f_{\sf d}(\cdot)$ from an arbitrary sample and denoted by $\mathbf{X}$. The first entry in $\mathbf{D}_{i}$, $\mathbf{X}_{(i)}$, is a tensor of arbitrary partial feature maps drawn from $\mathbf{X}$, which represent the feature maps already received by the server. The second entry $\mathcal{S}_{(i)}$ is an admissible subset of  indexes representing feature maps to be transmitted (hence not in $\mathbf{X}_{(i)}$). Then the label $H_{i}$ is the exact inference uncertainty generated by the server sub-model $f_{\sf s}(\cdot)$ for the input of a tensor $\mathbf{X}_{(i)}^{\prime}$ comprising both feature maps in $\mathbf{X}_{(i)}$ and the feature maps indexed in $\mathcal{S}_{(i)}$ drawn from $\mathbf{X}$, i.e., $H_{i}=-\mathbb{E}_{\ell}\left[\Pr{\left(\ell|\mathbf{X}_{(i)}^{\prime}\right)}\log\Pr{\left(\ell|\mathbf{X}_{(i)}^{\prime}\right)}\right]$.  Next, to train the prediction model, the loss function is designed  to be the mean-square error between the predicted result  and the ground-truth. A \emph{stochastic gradient descent} (SGD) optimizer is adopted to train the model, denoted as $f_{\sf r}(\cdot\vert \mathbf{W}_{\sf r})$ with parameters $\mathbf{W}_{\sf r}$,  by minimizing  the loss function: 
\begin{equation}
\min_{\mathbf{W}_{\sf r}} \frac{1}{|\mathcal{D}|}\sum_{i=1}^{|\mathcal{D}|}\left[ f_{\sf r}(\mathbf{D}_{i}|\mathbf{W}_{\sf r}) - H_{i} \right]^2.
\end{equation}
Given the current state $\theta_1$ and  the trained inference uncertainty predictor $f_{\sf r}(\cdot|\mathbf{W}_{\sf r})$, the online stopping control problem is written as 
\begin{equation}
\label{eqn: stopping_control_cnn}
\min\limits_{k^{*} \in \{0,1,...,K\} } \ f_{\sf r}\left( \left(\mathbf{X}_1,\bigcup_{k=1}^{k^{*}}\mathcal{S}_{k}^{\star}\right) \bigg| \mathbf{W}_{\sf r}\right) + c_0  k^{*}.
\end{equation}
The optimal stopping time from the current slot, namely  $k^{\star}=\!\!\mathop{\arg\min}\limits_{k^{*}\in\{0,1,..,K\}} f_{\sf r}\left(\! \left(\mathbf{X}_1,\bigcup\limits_{k=1}^{k^{*}}\mathcal{S}_{k}^{\star}\right) \bigg| \mathbf{W}_{\sf r}\right) + c_0  k^{*}$,  can be determined by linear search. Then  the stopping decision in the current slot   is $b_1^{\star}=\min\{1,k^{\star}\}$. Even for a Gaussian channel, the number of transmission slots for different samples differ due to their requirements of feeding different numbers of features into classifier to reach the same target confidence level if it is possible.

\begin{figure}
    \centering
    \includegraphics[height=3cm]{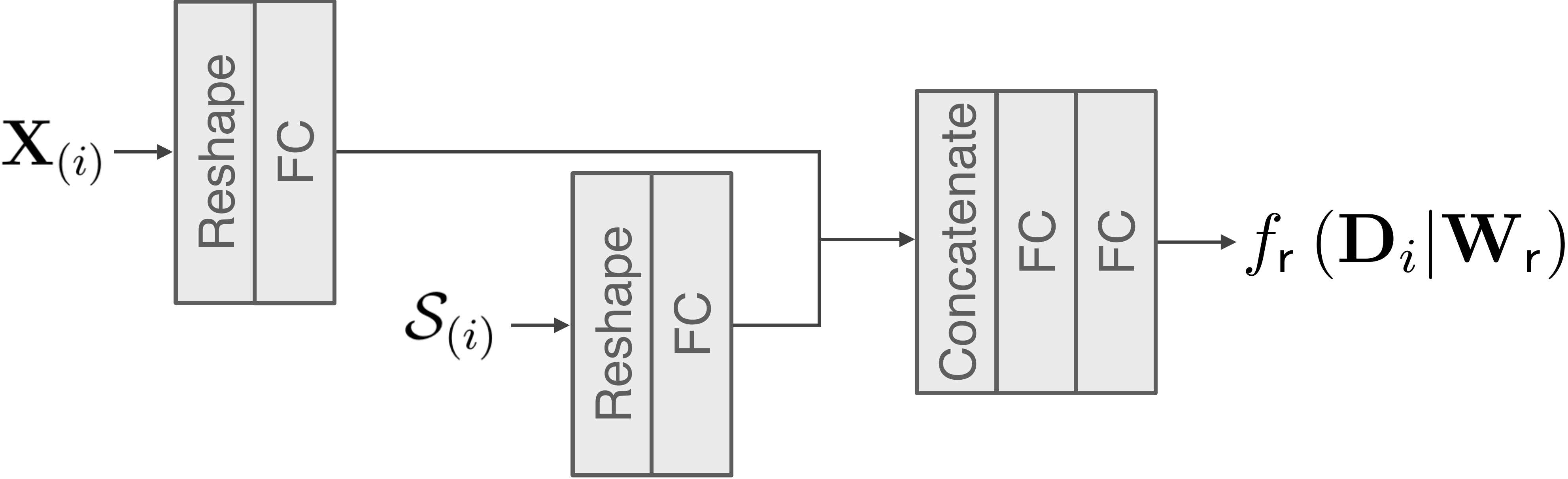}\vspace{-4mm}
    \caption{Illustration of the architecture for regression models.} 
    \label{fig: regression_model}\vspace{-5mm}
\end{figure}

\section{Experimental Results} \label{sec:  experiments}

\subsection{Experimental Settings}
The experimental setups are designed as follows, unless specified otherwise.  Each feature is quantized at a high resolution,  $Q=64$ bits/feature,  for digital transmission. The horizon in online scheduling is set as $K=5$ slots. A statistical (GM) dataset is used for training and testing the linear model and the popular MNIST  dataset of handwritten digits  for the CNN model. The  settings in  the two cases are described as follows. Consider the case of linear classification. For the GM dataset, there  are $L=2$ classes with $N=40$ features in total for each sample. Due to the small number of features, only a  narrow-band  channel is needed whose  bandwidth is set as $B=20$ kHz. The slot duration is  $T=10$ milliseconds. For the case of  Gaussian channel, the  channel SNR is set as $4$ dB such that the transmission rate is $Y_0=5$ features/slot. For the case of fading channel, the transmit rate is fixed at $Y_0=5$ features/slot while the outage probability is $p_{\sf o}=0.1$. For the case of  CNN, we use the well-known LeNet \cite{Molchanov2019CVPR}. The Gaussian channel in this case has a  bandwidth of  $B=2.6$ MHz,  the slot duration of $T=10$ milliseconds, and the channel SNR  $=4$ dB. The corresponding transmission rate is $Y_0=4$ feature-maps/slot.

The optimal control of ProgressFTX in the case of CNN classification  requires an uncertainty predictor as designed in Section~\ref{sec: cnn_model}.  Its architecture comprises two input layers, as illustrated in Fig. \ref{fig: regression_model}. The first one reshapes an input tensor for partial feature maps into a $512\times 1$ vector. The second one reshapes the feature selection input into a $256\times 1$ vector, where the coefficient is $1$ if the corresponding feature map index is selected or $0$ otherwise. These two vectors are concatenated and then fed into three fully-connected  layers with $100$, $40$, and $10$ neurons, respectively. There is one neuron in the last output layer providing a prediction of uncertainty. The test mean-square error  of the uncertainty predictor trained for $50$ epochs is low to $0.1$.

Two benchmarking schemes are considered. The first scheme, termed \emph{one-shot compression}, uses the classic approach of model compression: given importance-aware feature selection, the number of features to transmit, $Y_0 k^{\star}$, is determined prior to transmission such that it meets an uncertainty requirement ${H_0}.$ The value of $k^{\star}$ is solved from $ \mathop{\mathrm{argmin}}_{k^{*}} \tilde{H}(0, \mathrm{G}^{\star}(\theta_1,k^{*})) \geq{H_0}$ for linear classifiers or $ \mathop{\mathrm{argmin}}_{k^{*}} f_{\sf r}\left(\left(\mathbf{0},\bigcup_{k=1}^{k^*}\mathcal{S}_k^{\star}\right)\big\vert \mathbf{W}_{\sf r}\right) \geq{H_0}$ for CNN classifiers, respectively. The drawback of this scheme is that its lack of ACK/NACK feedback as for ProgressFTX makes the device inept in minimizing the number of features to achieve an exact uncertainty level. In the current experiments, the device has to instead target the expected uncertainty approximated in~\eqref{eqn: tilde_H} for the case of linear classification or predict the uncertainty in the case of CNN classification. The second scheme, termed \emph{random-feature optimal stopping}, modifies the optimal ProgressFTX by removing feature importance  awareness and instead selecting features randomly.

\subsection{Linear Classification Case}

First, we evaluate the effect of importance-aware feature selection  on inference performance. To this end, discriminant gains of different  feature dimensions, which are arranged in a decreasing order,  are plotted in Fig.~\ref{fig: linear_intuition}(a). In Fig.~\ref{fig: linear_intuition}(b), the levels of inference accuracy   and uncertainty  are plotted against the number of received features arranged in the same order. One can observe the monotonic reduction of accuracy and growth of uncertainty as the number of features increases. 
The result confirms the usefulness of importance awareness in feature selection for shortening the communication duration given target inference accuracy (or  expected uncertainty). 

\begin{figure}[t]
\centering
\subfigure[Discriminant Gains]{\includegraphics[height=5.45cm]{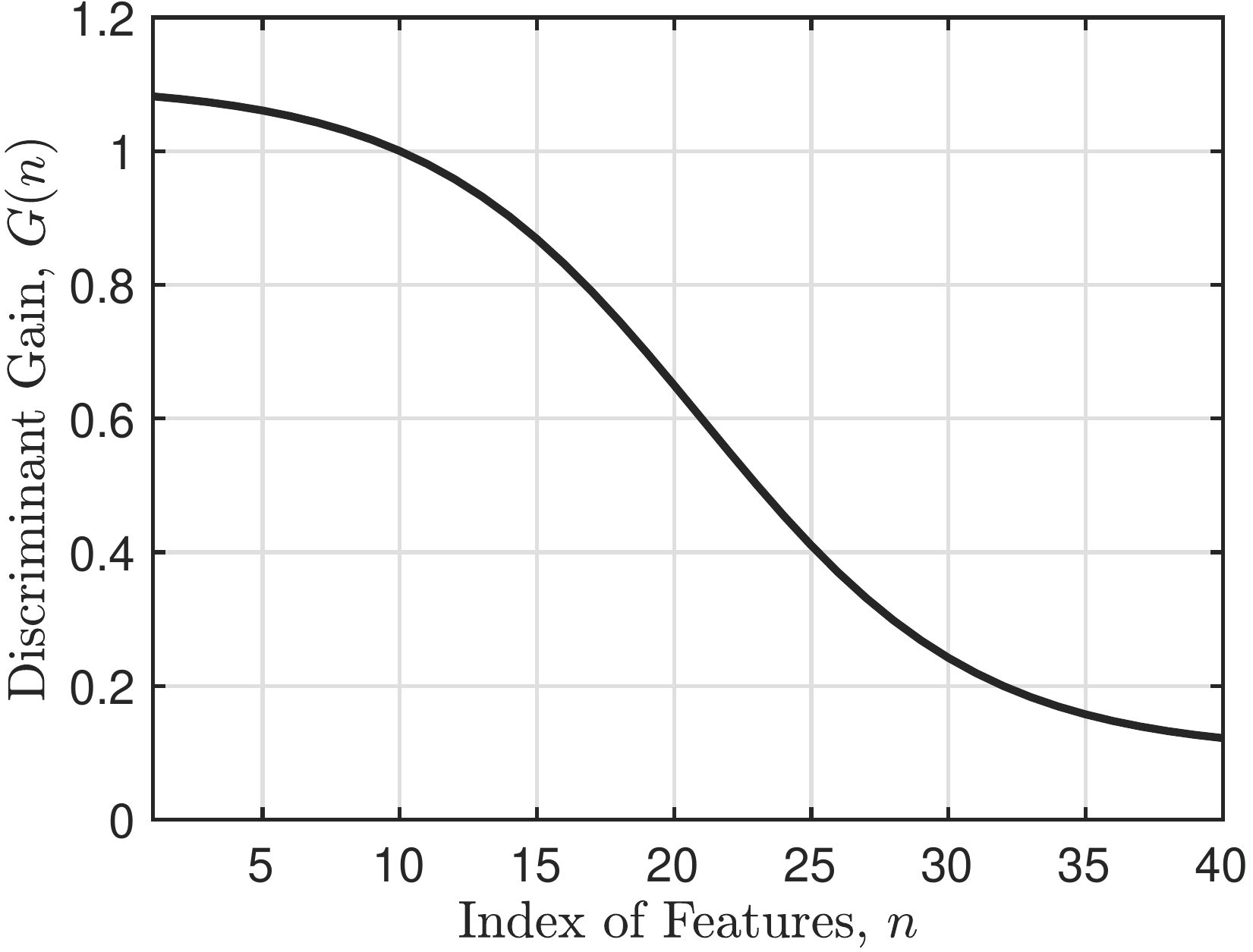}}
\hspace{0.2cm}
\subfigure[Inference Accuracy and Uncertainty]{\includegraphics[height=5.5cm]{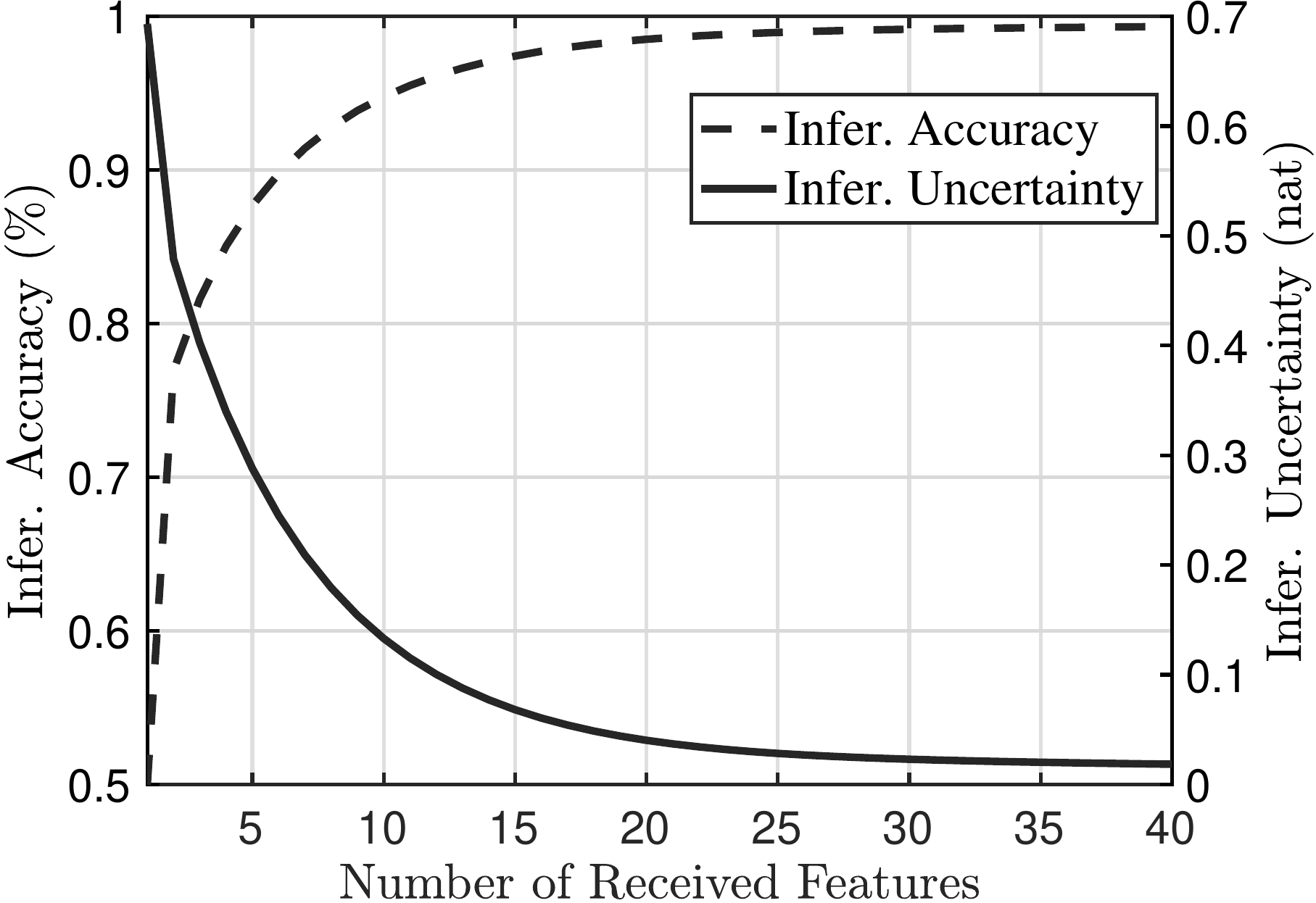}}
\vspace{-2mm}
\caption{{
(a) The discriminant gains of features and (b) their effects on inference performance (i.e., inference accuracy and uncertainty) in the case of linear classification.}}
\label{fig: linear_intuition}
\vspace{-4mm}
\end{figure}

Define the \emph{average communication latency}   of a transmission scheme for edge inference as the average number of transmission slots required for meeting a requirement on inference accuracy  or expected  uncertainty. Considering a Gaussian channel, the average communication latency of ProgressFTX and two benchmarking schemes are compared for varying target accuracy and  expected uncertainty  in Fig.~\ref{fig: linear_efficiency}. As observed from the figures, the proposed ProgressFTX technique achieves the lowest latency  based on both the criteria of achieving targeted uncertainty and  accuracy. For instance, two benchmarking schemes  require in average $2.2$ slots to achieve an accuracy of $95\%$  while ProgressFTX only requires $1.4$. In terms of  uncertainty, in average $2.8$ slots are used  to achieve uncertainty of $0.05$ for ProgressFTX. In comparison, the average latency  for one-shot compression and random-feature optimal stopping are about $39\%$ and $50\%$ higher, respectively. Moreover, ProgressFTX achieves the lowest average inference uncertainty of $0.018$ and the highest accuracy of $99\%$. The average communication latency of the three schemes are also compared 
in the case of a fading channel in  Fig.~\ref{fig: linear_efficiency_fading}. In the presence of fading, ProgressFTX continues to outperform  benchmarking schemes. The above comparisons demonstrate  the gains of importance awareness in feature selection and optimal stochastic control of transmission. 

\begin{figure}[t]
\centering
\subfigure[Inference Uncertainty]{\includegraphics[height=5.5cm]{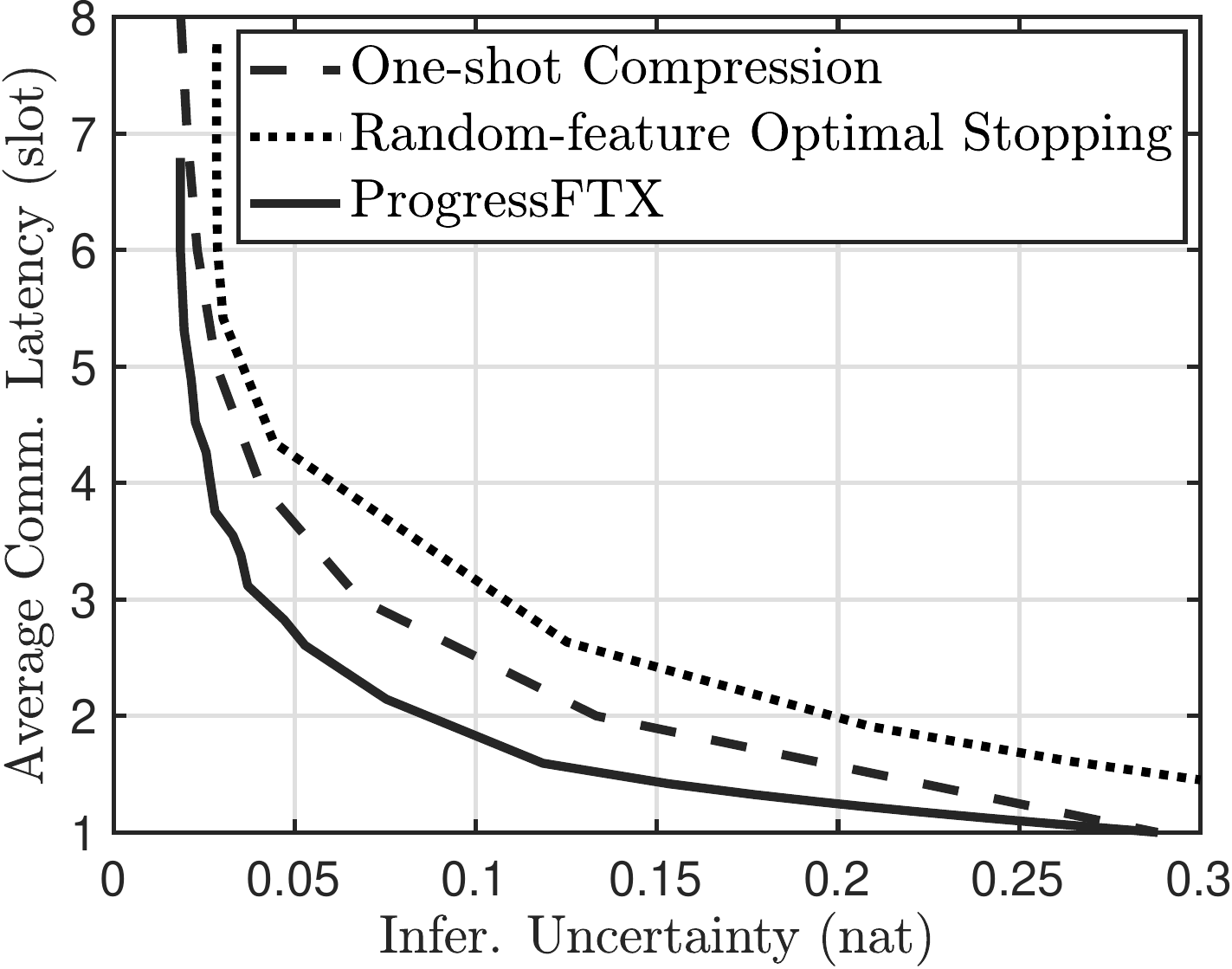}}
\hspace{1cm}
\subfigure[Inference Accuracy]{\includegraphics[height=5.5cm]{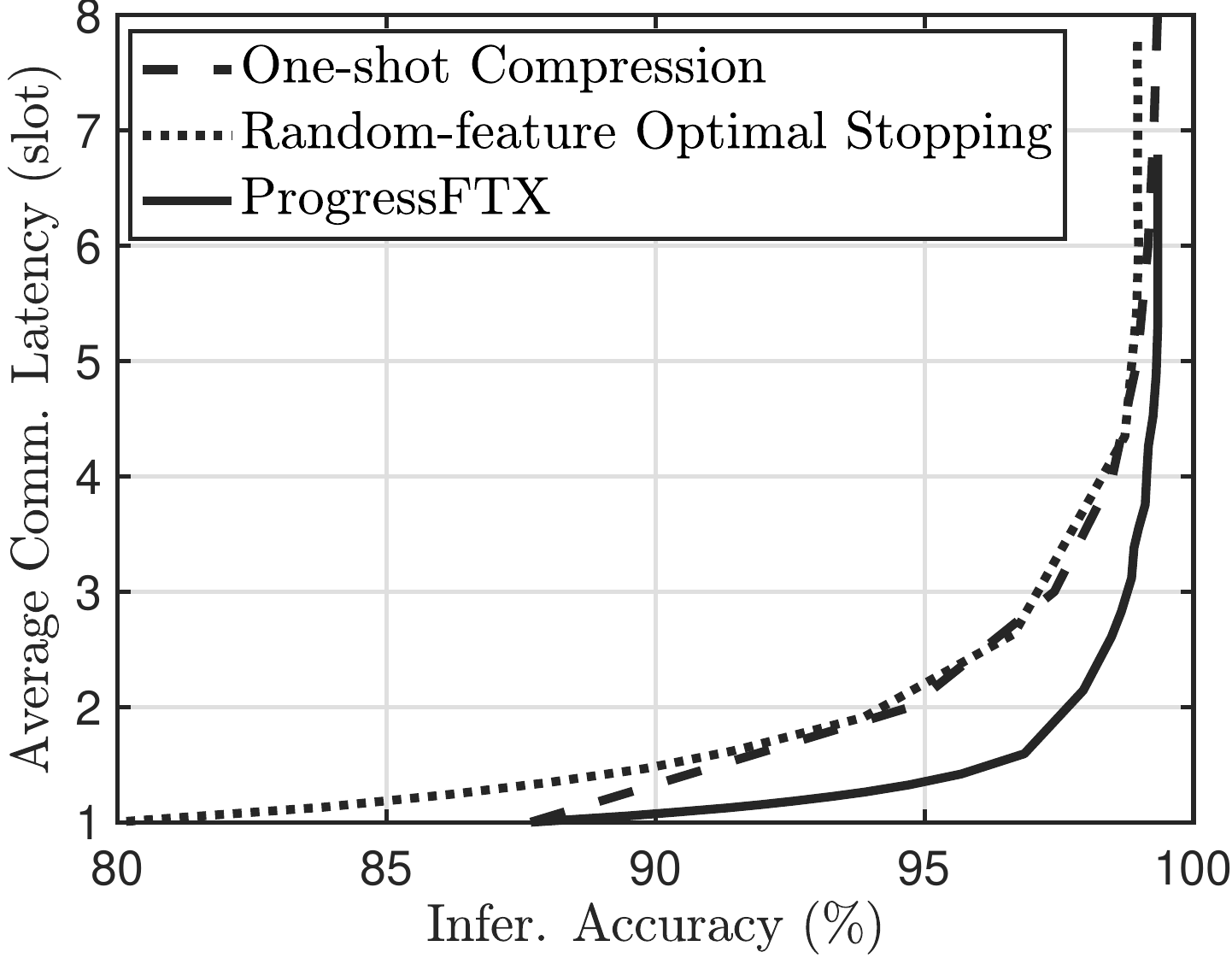}}
\vspace{-2mm}
\caption{ Comparison of average communication latency between ProgressFTX and benchmarking schemes in the case of linear classification and Gaussian channel for (a)  inference uncertainty or (b) inference accuracy.}
\label{fig: linear_efficiency}
\end{figure}
\begin{figure}[t]
\centering
\subfigure[Inference Uncertainty]{\includegraphics[height=5.5cm]{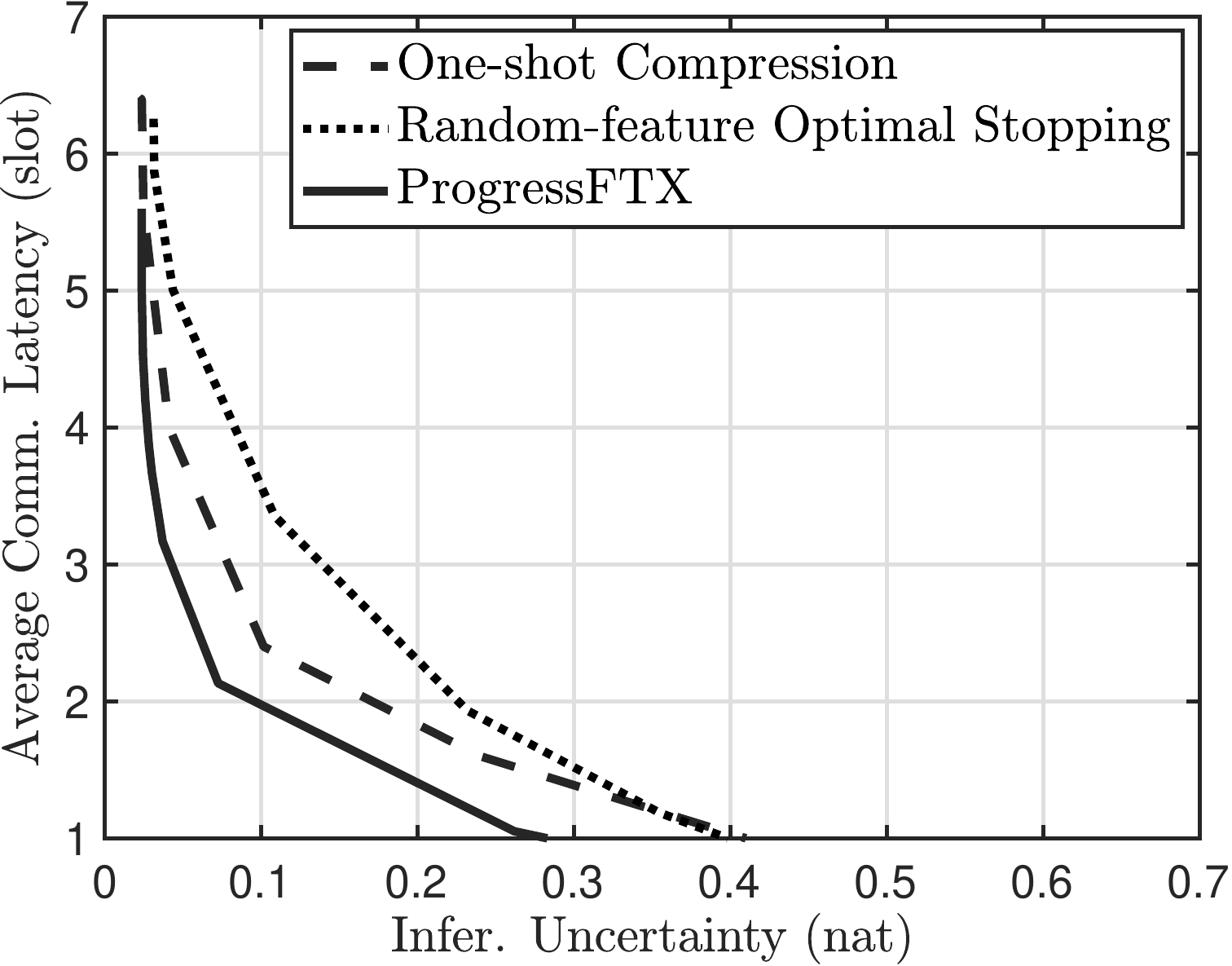}}
\hspace{1cm}
\subfigure[Inference Accuracy]{\includegraphics[height=5.5cm]{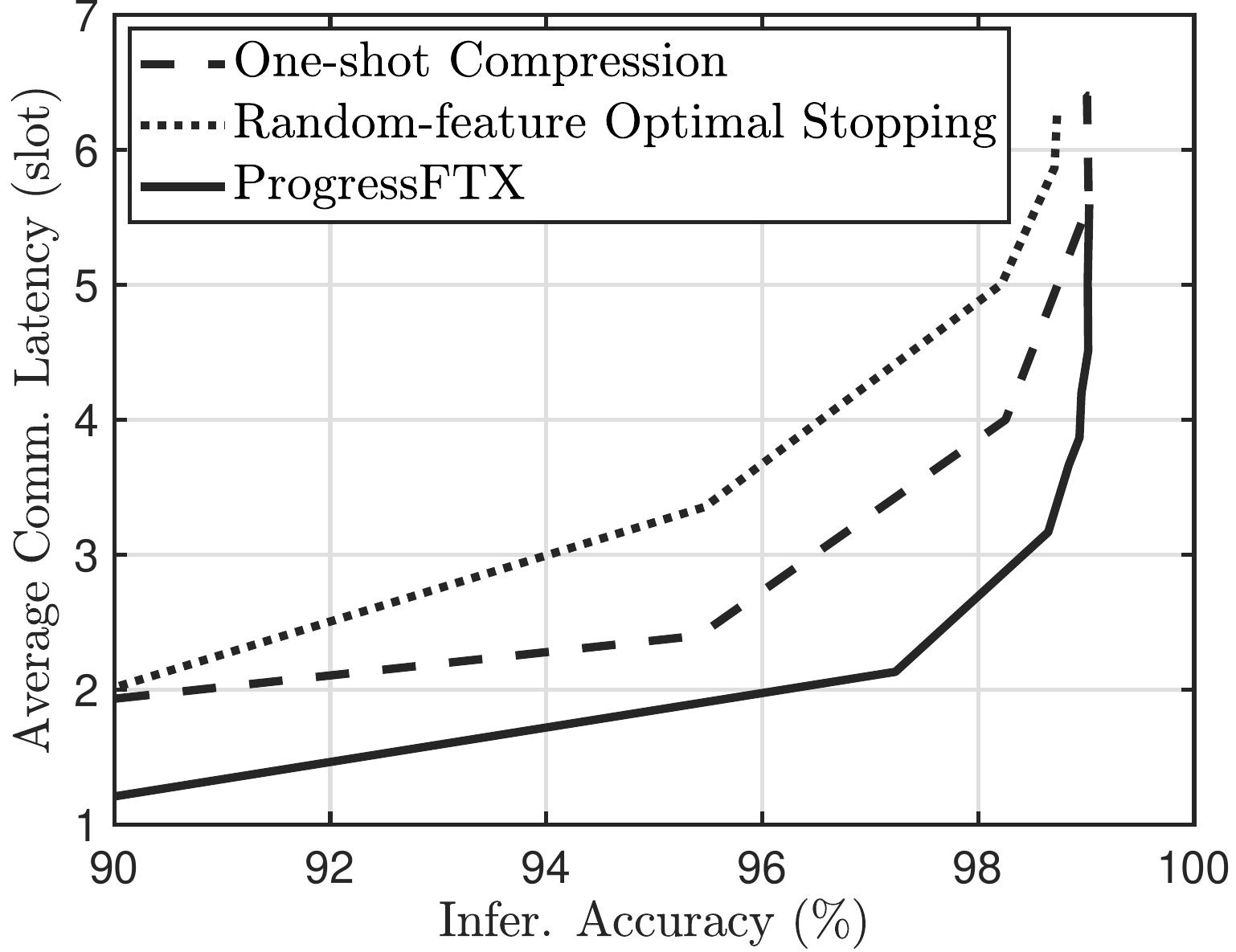}}
\vspace{-2mm}
\caption{Comparison of average communication latency between ProgressFTX and benchmarking schemes in the case of linear classification and fading channel with  varying  target (a) inference uncertainty or (b) inference accuracy.}
\label{fig: linear_efficiency_fading}
\vspace{-3mm}
\end{figure}

Define the \emph{transmission probability} of a feature dimension as the fraction of samples whose inference requires the transmission of the feature in this dimension.  To further compare ProgressFTX with  benchmarking schemes, the curves of  transmission probability of each feature dimension versus its importance level (i.e., discriminant gain) are plotted in Fig.~\ref{fig: linear_histogram} for the case of linear classification and Gaussian  channel. One can observe that the transmission probability is almost \emph{uniform}  over unpruned  feature dimensions for the scheme of one-shot compression and all dimensions for the scheme of random-feature optimal stopping. This indicates their lack of feature importance awareness. In contrast, the probabilities for ProgressFTX are highly skew with higher  probabilities for more important feature dimensions and vice versa. The skewness arises from the importance-aware feature selection as well as the stochastic control of transmissions which are the key reasons for the performance gain of ProgressFTX over the benchmark schemes. 

\begin{figure}[t]
\centering
\includegraphics[width=0.98\textwidth]{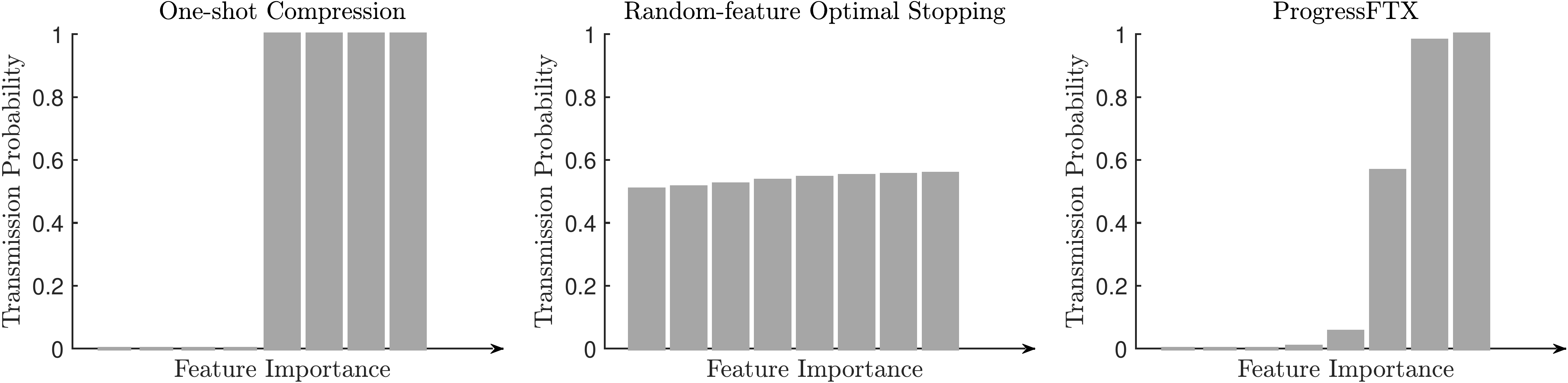}
\vspace{-5mm}
\caption{The transmission probabilities   of feature dimensions for the case of  linear classification and  Gaussian channel with  target inference accuracy of $98.5\%$. 
}
\label{fig: linear_histogram}
\vspace{-5mm}
\end{figure}

\subsection{CNN Classification Case}
\label{subsec: cnn_experiments}

\begin{figure}[t]
\centering
\includegraphics[height=5.45cm]{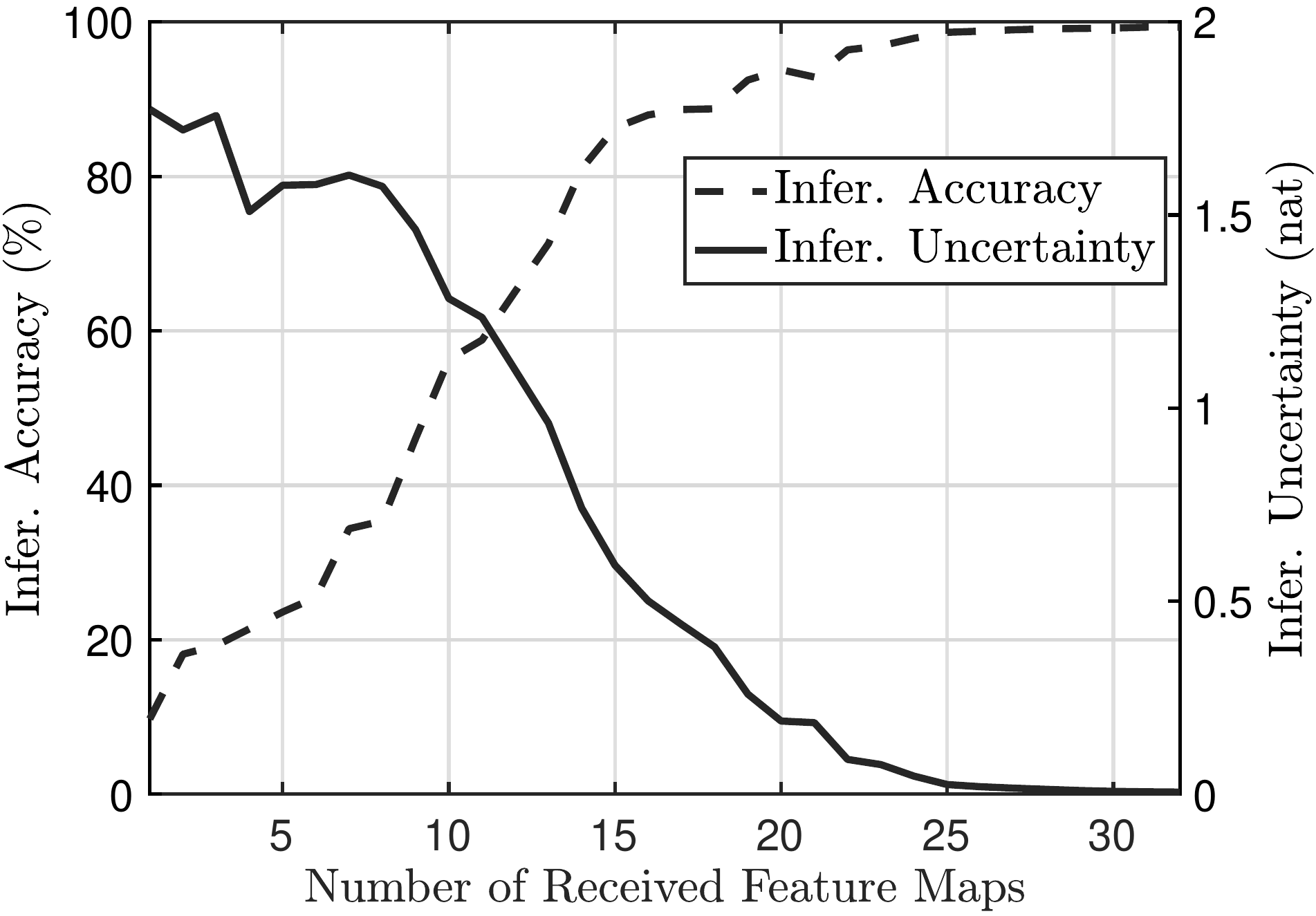}
\vspace{-6mm}
\caption{Effect of the number of received feature maps on CNN classification accuracy with importance-aware feature selection.}
\label{fig: cnn_intuition}
\vspace{-5mm}
\end{figure}

\begin{figure}[t]
\centering
\subfigure[Inference Uncertainty]{\includegraphics[height=5.6cm]{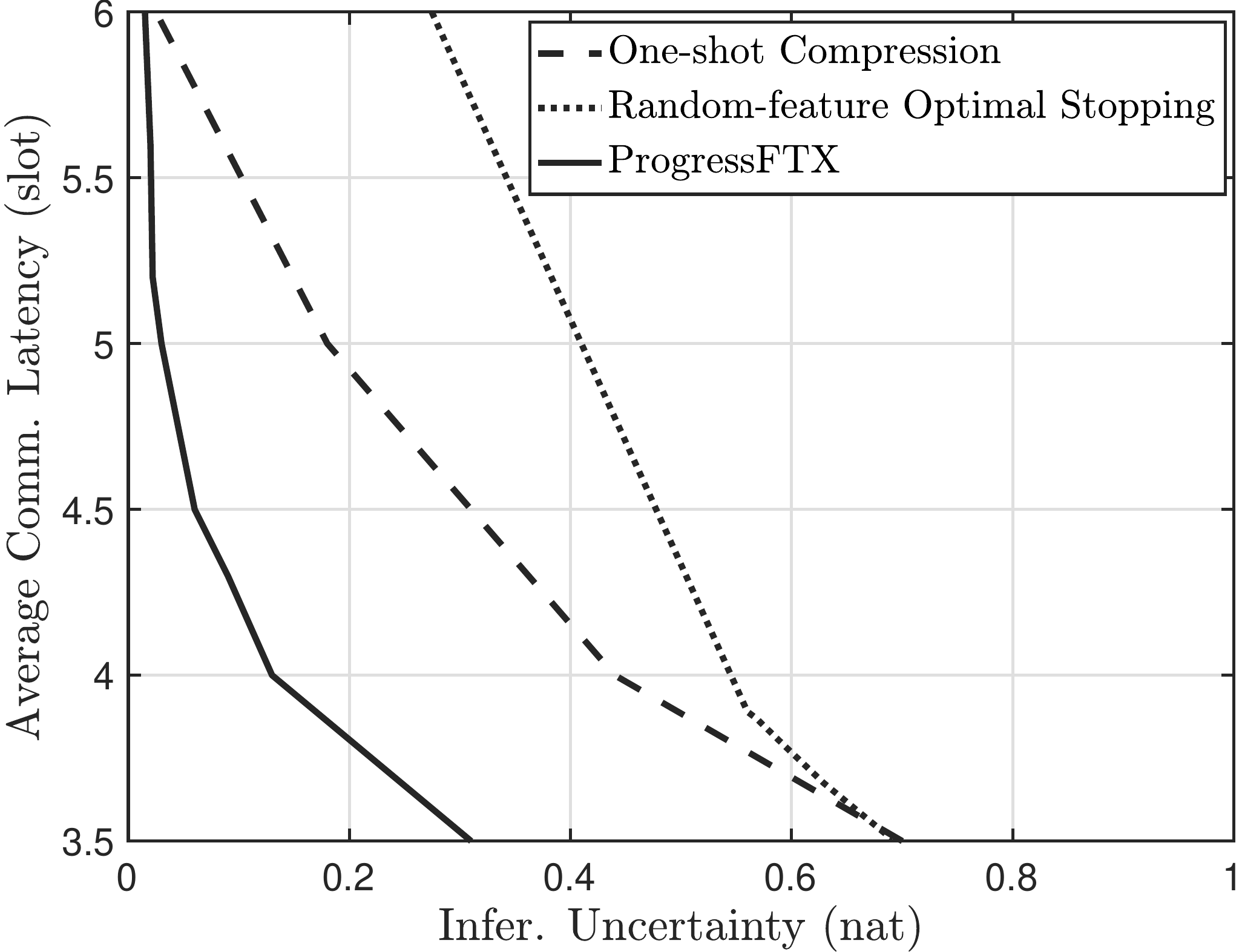}}
\hspace{1cm}
\subfigure[Inference Accuracy]{\includegraphics[height=5.6cm]{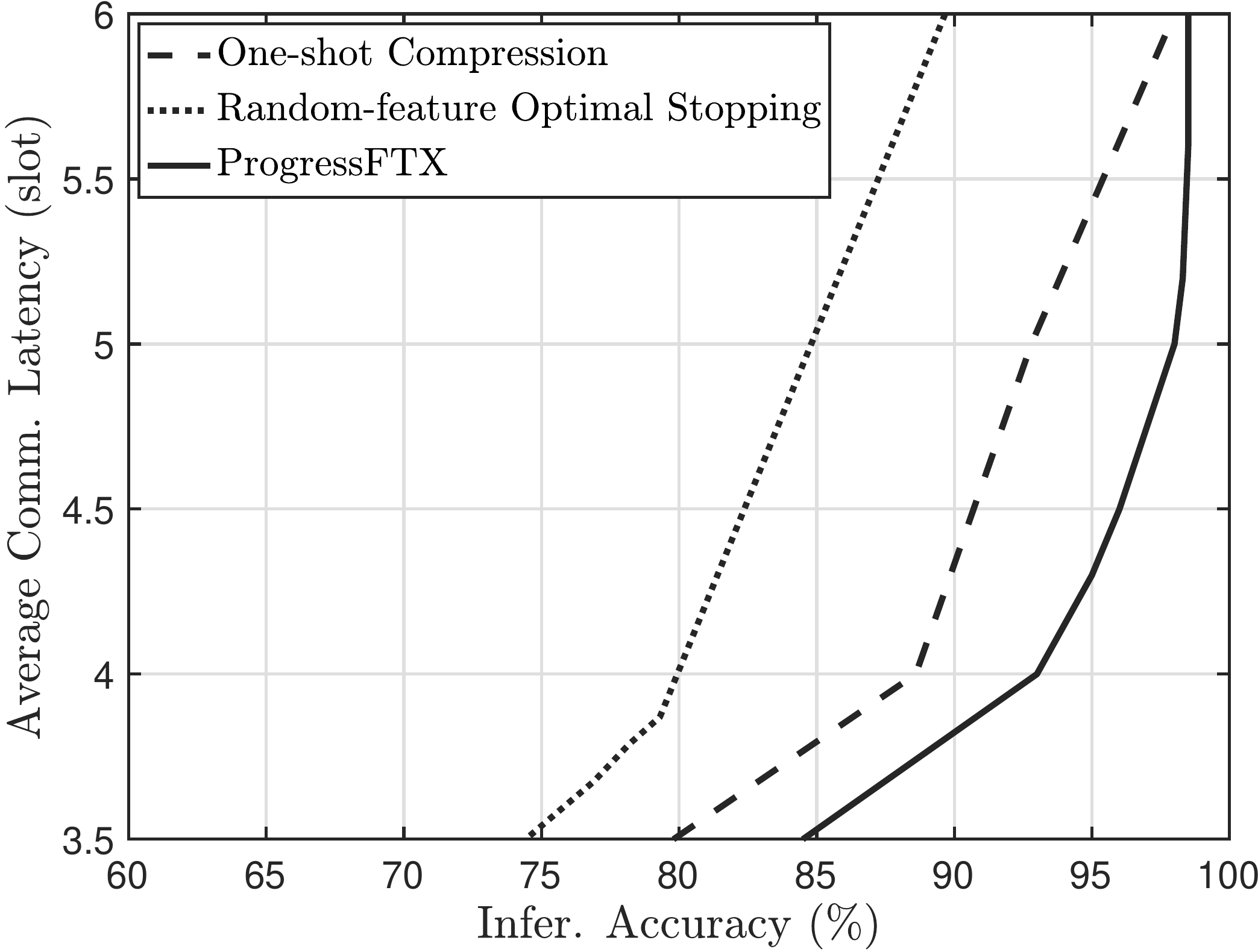}}
\vspace{-2mm}
\caption{Comparison of average communication latency  between ProgressFTX and benchmarking schemes for the case of CNN classification and  Gaussian channel  with 
varying  target (a) inference uncertainty or (b) inference accuracy.}
\label{fig: lenet_efficiency}
\vspace{-6mm}
\end{figure}

We consider the case of CNN classification where  ProgressFTX is controlled using the  algorithms developed in Section~\ref{sec: cnn_model}. To implement split inference, the split point of LeNet  is chosen to be right after the second CONV layer in the model. As a result, the device can choose from $32$ $4\times 4$ feature maps for transmission to the server. To rein in the complexity of model architecture and training, we consider a Gaussian channel as commonly assumed in the literature (see e.g.,~\cite{Niu2019Infocom,Deniz2020SPAWC}). Given these settings, one sample per inference task and $Y_0=4$ feature maps transmitted in one slot, at most $8$ slots are required to transmit all maps.

First, the curves of inference performance (i.e., inference accuracy and uncertainty) versus the number of transmitted feature maps, selected with importance awareness,  are plotted in Fig.~\ref{fig: cnn_intuition}. The tradeo-ffs are similar to those  in  the linear-classification case in Fig.~\ref{fig: linear_intuition}(b). Particularly, due to importance-aware feature selection, the accuracy is observed to rapidly grow and the uncertainty quickly reduces as the number of received feature maps increases.

The comparisons in Fig.~\ref{fig: linear_efficiency} is repeated in Fig.~\ref{fig: lenet_efficiency} for the CNN classifier. The same observation can be made that ProgressFTX outperforms the benchmark schemes over the considered ranges of inference uncertainty and accuracy. For instance, given $4$ transmission slots in average, ProgressFTX achieves the uncertainty of $0.13$ and accuracy of $93\%$ which are at least $65\%$ lower and $6.9\%$ higher than the benchmarking schemes. As another example, for the target accuracy of $93\%$, one can observe from   Fig.~\ref{fig: lenet_efficiency}(b) that ProgressFTX requires in average $4$ slots while one-shot compression requires  $5$, corresponding to  $20\%$  latency reduction for the former. 
Last, the transmission probabilities of different  feature maps are plotted against their importance levels  in Fig.~\ref{fig: lenet_insight}. The observations are similar to those for the linear model in Fig.~\ref{fig: linear_histogram}.

\begin{figure}[t]
    \centering
    \includegraphics[height=4cm]{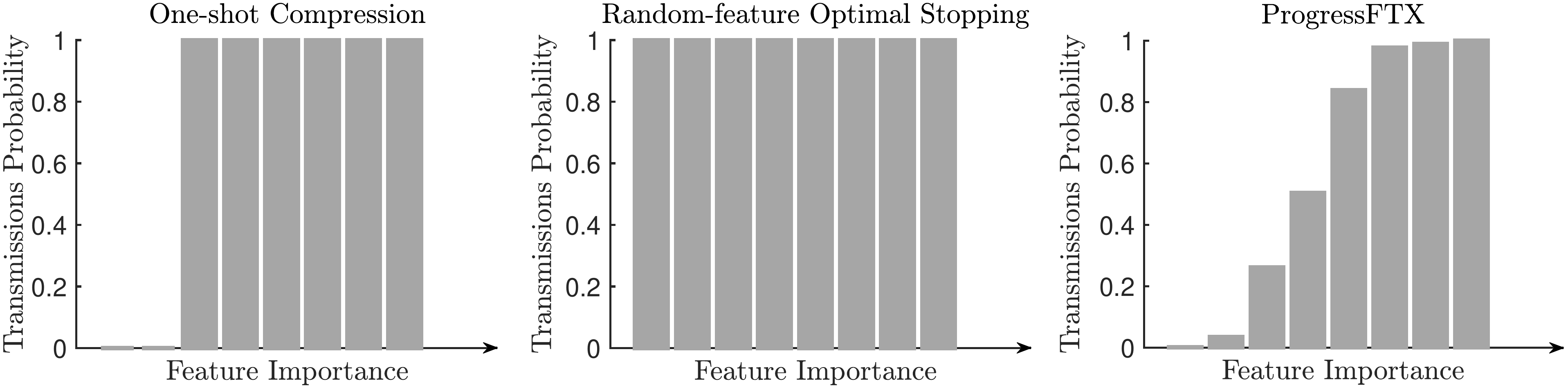}
    \vspace{-4mm}
    \caption{The transmission probabilities   of feature maps for the case of  CNN classification and  Gaussian channel with  target inference accuracy of $98.5\%$.}
    \label{fig: lenet_insight}
\end{figure}

\section{Concluding Remarks}
\label{sec: conclusions}
In this paper, a novel feature transmission technique termed ProgressFTX has been proposed for communication-efficient edge inference. ProgressFTX progressively selects and transmits important features until the desired inference performance level is reached, or terminates the transmission judiciously to improve the inference performance while reducing the transmission cost. Comprehensive experimental results demonstrate its benefits in reducing the communication latency under given inference performance targets. This first study of ProgressFTX opens a new direction to edge inference. A series of interesting topics for further study are warranted, such as power control and device scheduling for multi-device ProgressFTX.

\appendix
\subsection{Proof of Lemma~\ref{lemma: exponential_bound}}
\label{appendix: lemma_exponential_bound}
Substituting~\eqref{eqn: analysis_binary_pdf_delta} and~\eqref{eqn: h_ub_definition} into~\eqref{eqn: ub_tight_binary}, one can observe that the targeted integral consists of a group of integrals of Gaussian and exponential functions. This kind of integral is well-known in the literature that not to tend itself for a closed-form expression but can be expressed in terms of multiple \emph{complementary error functions}~\cite{Olabiyi2012WCL,Chiani2003TWC}. 
Due to limited space, we do not present the expression of $\bar{H}^{\sf ub}(\delta_1,G(\mathcal{S}))$
but directly report a fact derived from the expression, 
$\lim\limits_{G(\mathcal{S})\to \infty}\frac{\bar{H}^{\sf ub}(\delta_1,G(\mathcal{S}))}{{\left(e^{\frac{\delta_1}{2}}+e^{-\frac{\delta_1}{2}}\right)}\left(G(\mathcal{S})\right)^{-\frac{1}{2}}e^{-\frac{G(\mathcal{S})}{8}} } = \frac{16}{9}\sqrt{\frac{2}{\pi}}$.
This results in $\bar{H}^{\sf ub}(\delta_1,G(\mathcal{S}))=O\left({\left(e^{\frac{\delta_1}{2}}+e^{-\frac{\delta_1}{2}}\right)}\left(G(\mathcal{S})\right)^{-\frac{1}{2}}e^{-\frac{G(\mathcal{S})}{8}}\right)$, where $O(\cdot)$ is the Bachmann–Landau notation establishing equivalence between shrinking rates of two functions. Consequently, given any $\delta_1$, $\bar{H}^{\sf ub}(\delta_1,G(\mathcal{S}))$ \emph{monotonically} decays to zero at a \emph{super-exponential} rate in the asymptotic regime of $G(\mathcal{S})$. This scaling law proofs Lemma~\ref{lemma: exponential_bound}. 

\subsection{Proof of Convexity of Problem (P5)}
\label{app: proof_convex_stopping_fading}

To show the convexity of Problem (P5), it is equivalent to prove the following inequality holds for $k^{*}\geq1$, that is
$
\Phi(\theta_1,k^{*}-1)+\Phi(\theta_1,k^{*}+1)-2\Phi(\theta_1,k^{*})\triangleq {\nabla^2 \Phi(\theta_1,k^{*})} \geq0,
$
where $\nabla^2 \Phi(\theta_1,k^{*})$ is the second-order difference of $\Phi(\theta_1,k^{*})$ w.r.t. its second argument. To begin with, we construct a function $\tilde{\Phi}(k^{*})$ defined as
$
\tilde{\Phi}(k^{*}) \triangleq \sum_{k^{\prime}=0}^{k^{*}}{ \left(l_1 k^{\prime} + l_2 \right) p(k^{\prime}|k^{*}) }, 
$
 where $l_1 k^{\prime} + l_2$ is an \emph{arbitrary} linear function with constants $l_1$ and $l_2$. The closed-form expression of $\tilde{\Phi}(k^{*})$ is evaluated to be  $\tilde{\Phi}(k^{*})=l_1 k^{\prime}(1-p_{\sf o}) + l_2$. Consequently, the second-order difference of $\tilde{\Phi}(k^{*})$ w.r.t. $k^{*}$ always equals to zero, i.e.,
\begin{equation}
    \nabla^2\tilde{\Phi}(k^{*})=0, \ \forall k^{*}. \label{eqn: diff2_zero}
\end{equation}
The analytical form of the second-order difference of PMF w.r.t. $k^{*}$, which is denoted as $\nabla^2 p(k^{\prime}|k^{*})=p(k^{\prime}|k^{*}-1) + p(k^{\prime}|k^{*}-1) - 2p(k^{\prime}|k^{*})$, is given by
\begin{eqnarray}
\nabla^2 p(k^{\prime}|k^{*})
&=& p_{\sf o}^{k^{*}-k^{\prime}}(1-p_{\sf o})^{k^{\prime}}\frac{\prod_{i=k^{*}-k^{\prime}+2}^{k^{*}-1}i}{ \prod_{i=1}^{k^{\prime}}i} \nonumber\\
&& \times\underbrace{\left[p_{\sf o}^{-1}(k^{*}-k^{\prime}+2)(k^{*}-k^{\prime})+p_{\sf o}k^{*}(k^{*}+1)-2k^{*}(k^{*}-k^{\prime}+1)\right]}_{\text{denoted as $\eta\left(k^{\prime}\right)$}}.    \label{eqn: diff2_pmf}
\end{eqnarray}
Observes from~\eqref{eqn: diff2_pmf} that the sign of $\nabla^2 p(k^{\prime}|k^{*}) $ is determined by $\eta\left(k^{\prime}\right)$, 
a convex quadrature function of $k^{\prime}$. Specifically, $\eta\left(0\right)\geq 0$ and $\eta\left(k^{*}\right)\geq 0$ hold. We conclude that given $k^{*}$, $\nabla^2 p(k^{\prime}|k^{*}) $ is non-positive only in one closed integral of \emph{consecutive} integers $k^{\prime}$, where the left and right end-points are denoted by $k_1>0$ and $k_2<k^{*}$, respectively. 

We expand $\nabla^2{\Phi}(\theta_1,k^{*})$ and $\nabla^2\tilde{\Phi}(k^{*})$ as
$
\nabla^2\Phi(\theta_1,k^{*}) = \sum_{k^{\prime}=0}^{k^{*}} \tilde{H}\left(\delta_1,\mathrm{G}^{\star}(\theta_1,k^{\prime})\right) \nabla^2 p(k^{\prime}|k^{*})
$
 and
$
\nabla^2\tilde{\Phi}(k^{*}) = \sum_{k^{\prime}=0}^{k^{*}} \left(l_1 k^{\prime} + l_2 \right) \nabla^2 p(k^{\prime}|k^{*}), 
$
respectively. 
Choose $l_1\leq 0$ and $l_2\geq 0$ such that $l_1 k^{\prime} + l_2 > 0$ if $k^{\prime}\geq k_1 $ and $l_1 k^{\prime} + l_2 \geq 0$ if $k^{\prime} > k_1 $. Due to $\tilde{H}\left(\delta_1,\mathrm{G}^{\star}(\theta_1,k^{\prime})\right)\geq 0$, one can conclude that 
\begin{equation}
    \label{ineq: diff2_compare_right}
    \sum_{k^{\prime}=k_2+1}^{k^{*}} \tilde{H}\left(\delta_1,\mathrm{G}^{\star}(\theta_1,k^{\prime})\right) \nabla^2 p(k^{\prime}|k^{*}) \geq 0 \geq \sum_{k^{\prime}=k_2+1}^{k^{*}} \tilde{H}\left(\delta_1,\mathrm{G}^{\star}(\theta_1,k^{\prime})\right) \nabla^2 p(k^{\prime}|k^{*}).
\end{equation}
Next, we note that $\tilde{H}\left(\delta_1,\mathrm{G}^{\star}(\theta_1,k^{\prime})\right)$ is strictly convex and monotonically decreasing w.r.t. $k^{\prime}$ while $l_1 k^{\prime} + l_2$ is linear. It is always able to find an $l_1 \leq 0$ with sufficiently small $\vert l_1 \vert$ such that 
\begin{equation}
    \sum_{k^{\prime}=0}^{k_2} \tilde{H}\left(\delta_1,\mathrm{G}^{\star}(\theta_1,k^{\prime})\right) \nabla^2 p(k^{\prime}|k^{*}) \geq \sum_{k^{\prime}=0}^{k_2} \tilde{H}\left(\delta_1,\mathrm{G}^{\star}(\theta_1,k^{\prime})\right) \nabla^2 p(k^{\prime}|k^{*}) \geq 0.
    \label{ineq: diff2_compare_left}
\end{equation}
Combining \eqref{eqn: diff2_zero},~\eqref{ineq: diff2_compare_right} and~\eqref{ineq: diff2_compare_left} leads to the desired results, which completes the proof.
\begin{eqnarray}
    \nabla^2\Phi(\theta_1,k^{*}) \geq \nabla^2\tilde{\Phi}(k^{*}) = 0.
\end{eqnarray}

\bibliographystyle{IEEEtran}

\end{document}